\documentclass[a4paper,USenglish]{socg-lipics-v2019}
\usepackage{microtype}
\usepackage{mathtools}
\usepackage{algorithm}
\usepackage{algorithmicx}
\usepackage{algpseudocode}
\usepackage{graphicx}
\usepackage{comment}
\usepackage[normalem]{ulem}

\nolinenumbers

\DeclarePairedDelimiter\ceil{\lceil}{\rceil}
\DeclarePairedDelimiter\floor{\lfloor}{\rfloor}

\title{A Simple Algorithm for Higher-order Delaunay
     Mosaics and Alpha Shapes}
\titlerunning{A Simple Algorithm for Higher-order Delaunay
     Mosaics and Alpha Shapes}

\author{Herbert Edelsbrunner}{IST Austria (Institute of Science and Technology Austria),
  Am Campus 1, \\ 3400 Klosterneuburg, Austria}{edels@ist.ac.at}{https://orcid.org/0000-0002-9823-6833}{}
\author{Georg Osang}{IST Austria (Institute of Science and Technology Austria),
  Am Campus 1, \\ 3400 Klosterneuburg, Austria}{georg.osang@ist.ac.at}{https://orcid.org/0000-0002-8882-5116}{}
\authorrunning{H. Edelsbrunner and G. Osang}

\Copyright{Herbert Edelsbrunner and Georg Osang}
\ccsdesc[500]{Mathematics of computing~Combinatorial algorithms}

\keywords{Delaunay mosaics, Voronoi tessellations, algorithms,
  software, computational experiments.}

\funding{This project has received funding from the European Research
Council (ERC) under the European Union's Horizon 2020 research and
innovation programme, grant no.\ 788183, from the Wittgenstein Prize,
Austrian Science Fund (FWF), grant no.\ Z 342-N31, and from the DFG
Collaborative Research Center TRR 109, `Discretization in Geometry and
Dynamics', Austrian Science Fund (FWF), grant no.\ I 02979-N35.}

\newcommand{\mm}[1] {\ifmmode{#1}\else{\mbox{\(#1\)}}\fi}

\newcommand{\scalprod}[2] {{\langle #1 , #2 \rangle}}

\newcommand{\ignore}[1]{}

\newcommand{\ourproof}{\begin{proof}}
\newcommand{\eop}{\end{proof}}  %

\def \Xp{{A}}                     
\newcommand{\xp}           {{a}}  
\newcommand{\rp}           {{x}}  
\newcommand{\rpcd}         {{z}}  
\newcommand{\rpd}          {{y}}  
\newcommand{\vq}           {\mm{{v}}} 

\newcommand{\Rspace}       {\mm{{\mathbb R}}}

\newcommand{\Arr}[1]       {\mm{{\rm Arr}{({#1})}}}
\newcommand{\Rhomboid}[1]  {\mm{{\rm Rho}{({#1})}}}

\newcommand{\Del}[2]       {\mm{{\rm Del}_{#1}{({#2})}}}
\newcommand{\Vor}[2]       {\mm{{\rm Vor}_{#1}{({#2})}}}

\newcommand{\domain}[1]    {\mm{{\rm dom}{({#1})}}}

\newcommand{\Rfun}         {\mm{{\cal R}}}

\newcommand{\Ain}          {\mm{{\Xp}_{\it in}}}
\newcommand{\Aon}          {\mm{{\Xp}_{\it on}}}
\newcommand{\Aout}         {\mm{{\Xp}_{\it out}}}

\newcommand{\xin}[1]       {\mm{{\Xp}_{\it in}{({#1})}}}
\newcommand{\xon}[1]       {\mm{{\Xp}_{\it on}{({#1})}}}
\renewcommand{\xout}[1]      {\mm{{\Xp}_{\it out}{({#1})}}}

\newcommand{\vx}[1]        {\mm{{V}{({#1})}}}

\newcommand{\Smin}[1]      {\mm{{S}_{\it min}{({#1})}}}

\newcommand{\Paraboloid}   {\mm{{\mathcal P}}}  
\newcommand{\lift}[1]      {\mm{{\rm lift}{({#1})}}}  
\newcommand{\Paraboloidt}[1] {\mm{{\mathcal P}_{#1}}}  
\newcommand{\parafunct}[1] {\mm{{\mathcal P}_{#1}}}  

\renewcommand{\rho}        {\mm{\varrho}}

\newcommand{\norm}[1]      {\mm{\|{#1}\|}}
\newcommand{\card}[1]      {\mm{\#{#1}}}

\newcommand{\Edist}[2]     {\mm{\|{#1}-{#2}\|}}

\newcommand{\ourparagraph}[1] {\vspace{0.1in} \noindent \textbf{#1}}
\newcommand{\Skip}[1]      {}
\newcommand{\Shorten}[1]   {}

\begin{document}
\maketitle

\begin{abstract}
\nolinenumbers
  We present a simple algorithm for computing higher-order Delaunay mosaics
  that works in Euclidean spaces of any finite dimensions.
  The algorithm selects the vertices of the order-$k$ mosaic from incrementally constructed
  lower-order mosaics and uses an algorithm for weighted first-order Delaunay mosaics
  as a black-box to construct the order-$k$ mosaic from its vertices.
  Beyond this black-box, the algorithm uses only combinatorial operations,
  thus facilitating easy implementation.
  We extend this algorithm to compute higher-order $\alpha$-shapes
  and provide open-source implementations. 
  We present experimental results for properties of higher-order Delaunay mosaics
  of random point sets.
\end{abstract}
\nolinenumbers

\newpage
\setcounter{page}{1}
\section{Introduction}
\label{sec:intro}

Order-$k$ Voronoi tessellations are a generalization of
ordinary Voronoi tessellations.
Instead of each domain corresponding to a single point in the input,
$A \subseteq \Rspace^d$, each order-$k$ domain corresponds to a subset,
$Q \subseteq A$, of size $k$, and consists of the set of points
$x \in \Rspace^d$ for whom the points in $Q$ are the closest $k$ points within $A$.
Its dual is the order-$k$ Delaunay mosaic.
We will formally define both in Section \ref{sec:definitions}.
Order-$k$ Voronoi tessellations were introduced by \cite{ShHo75} as a data
structure for fast $k$ closest point queries,
namely in time $O(k + \log n)$ with $n = \card{A}$.
A less direct application is the computation of the distance-to-measure
introduced in \cite{chazal2011geometric} and related
to $k$ closest point search in \cite{GMM13}.
Furthermore, certain subcomplexes of the order-$k$ Delaunay mosaic realize
the order-$k$ $\alpha$-shapes introduced in \cite{KrPo14}.
Order-$k$ $\alpha$-shapes are a generalization of $\alpha$-shapes \cite{EKS83} used
to approximate the shape of a point set. Unlike ordinary $\alpha$-shapes and depending
on the parameter $k$, they exhibit robustness to noisy data points.

In the plane, the number of domains in the order-$k$ Voronoi tessellation
or, equivalently, the number of vertices in the order-$k$ Delaunay mosaic is
$\Theta(k(n - k))$; see \cite{Lee82,ShHo75}.
For dimensions $d \geq 3$, this number can vary significantly depending on the
way the input points are distributed.
The upper bound of $O(k^{\ceil{\frac{d+1}{2}}} n^{\floor{\frac{d+1}{2}}})$
on the total size of the first $k$ higher-order Delaunay mosaics
\cite{clarkson1989applications} is tight,
while the lower bound of $\Omega(k^d n)$ \cite{Mul90} is only conjectured.
For individual order-$k$ Delaunay mosaics, the complexity is poorly understood.
The problem is closely related to the $(d+1)$-dimensional $k$-set problem.
Specifically, the points in $A \subseteq \Rspace^d$ can be mapped to
equally many points in $\Rspace^{d+1}$ such that the order-$k$ domains
in $\Rspace^{d}$ correspond to $k$-sets in $\Rspace^{d+1}$, 
see e.g.\ \cite{clarkson1989applications}.

The first algorithm to compute order-$k$ Voronoi tessellations and Delaunay mosaics
in the plane was described by Lee in \cite{Lee82}.
The algorithm computes the Voronoi tessellations one by one,
in increasing order and in time $O(k^2 n \log n)$.
Mulmuley \cite{Mul90} extended this algorithm beyond two dimensions,
computing the first $k$ levels in a special $(d+1)$-dimensional hyperplane arrangement,
which implicitly yield the order-$k$ Voronoi tessellations and Delaunay mosaics
in time $O(s \log n + k^d n^2)$, in which $s$ denotes the output size.
Mulmuley \cite{mulmuley1991levels} later described another algorithm,
which instead adds hyperplanes one by one,
and runs in time $O(k^{\ceil{\frac{d+1}{2}}} n^{\floor{\frac{d+1}{2}}})$
for $d \geq 3$, which equals the worst-case output size.
For $d=2$, the expected runtime is $O(k^2 n \log \frac{n}{k})$.
Another incremental algorithm with similar complexity for $d \geq 3$ has
been described by Agarwal \emph{et al}.\ \cite{agarwal1998constructing}.

In this paper, we describe a new algorithm for computing order-$k$ Delaunay mosaics
in Euclidean space of any finite dimension that stands out in its simplicity.
It employs an algorithm for weighted first-order Delaunay mosaics,
and otherwise uses only combinatorial operations.
It thus benefits from highly optimized existing implementations and, if desired,
can build upon their use of exact arithmetic.
Its complexity depends on the complexity of the algorithm used for weighted
Delaunay mosaics. Assuming it is linear in its output size, then out algorithm
is also linear in its output size.
We implement this algorithm and run it on various point sets,
shedding light on the size and other properties of order-$k$ Delaunay mosaics.
In particular, we compare the total size of the first $k$ Delaunay mosaics of
random point sets with the (tight) worst-case upper bound,
and we study the size of individual order-$k$ Delaunay mosaics,
for which no tight bounds are known in general.
As far as we are aware,
no such experimental investigations have been performed in the past,
possibly due to the lack of a practical algorithm.
We extend our algorithm to compute the radius function
on an order-$k$ Delaunay mosaic,
which gives us the subcomplexes realizing order-$k$ $\alpha$-shapes.
Open-source implementations of our algorithm
are available \cite{orderkgithub,rhomboidgithub}.

Our algorithm makes use of the \emph{rhomboid tiling} \cite{EdOs18}, which we will
introduce in Section \ref{sec:definitions} alongside other necessary definitions.
We will explore the combinatorial properties of this tiling and, by proxy,
the properties of order-$k$ Delaunay mosaics in Section \ref{sec:combinatorial}.
Using these results, we explain our algorithm in Section \ref{sec:algorithm}.
We present experimental results obtained with two implementations of this algorithm
in Section \ref{sec:experimental}. Section \ref{sec:alphashapes} introduces
a radius function on the order-$k$ Delaunay mosaics and a way to compute it
to yield order-$k$ $\alpha$-shapes.
We close with a discussion of possible extensions
and optimizations in Section \ref{sec:discussion}.

\section{Definitions}
\label{sec:definitions}

Given a locally finite set, $A \subseteq \Rspace^d$,
the \emph{(Voronoi) domain} of $Q \subseteq A$ is
$\domain{Q} = \{x \in \Rspace^d \mid \Edist{x}{q} \leq \Edist{x}{a},
  \forall q \in Q, \forall a \in A \setminus Q \}$.
Its \emph{order} is $\card{Q}$.
For each positive integer $k$, the \emph{order-$k$ Voronoi tessellation} is
$\Vor{k}{A} = \{ \domain{Q} \mid Q \subseteq A, \card{Q} = k \}$.
The \emph{order-$k$ Delaunay mosaic} is the cell complex dual to $\Vor{k}{A}$,
denoted $\Del{k}{A}$.
To realize the mosaic geometrically, we usually use the average of the
points in $Q$ as the location of the corresponding vertex in $\Rspace^d$.
In a few instances we use the sum rather than the average, for convenience. 
Figure \ref{fig:order2_full} shows an example for $k=2$.
As we will see shortly, in $d \geq 3$ dimensions,
the order-$k$ Delaunay mosaic is not necessarily simplicial even if the points
are in general position.
\begin{figure}[h]
  \centering \vspace{0.1in}
  \includegraphics[width=0.6\textwidth, page=1]{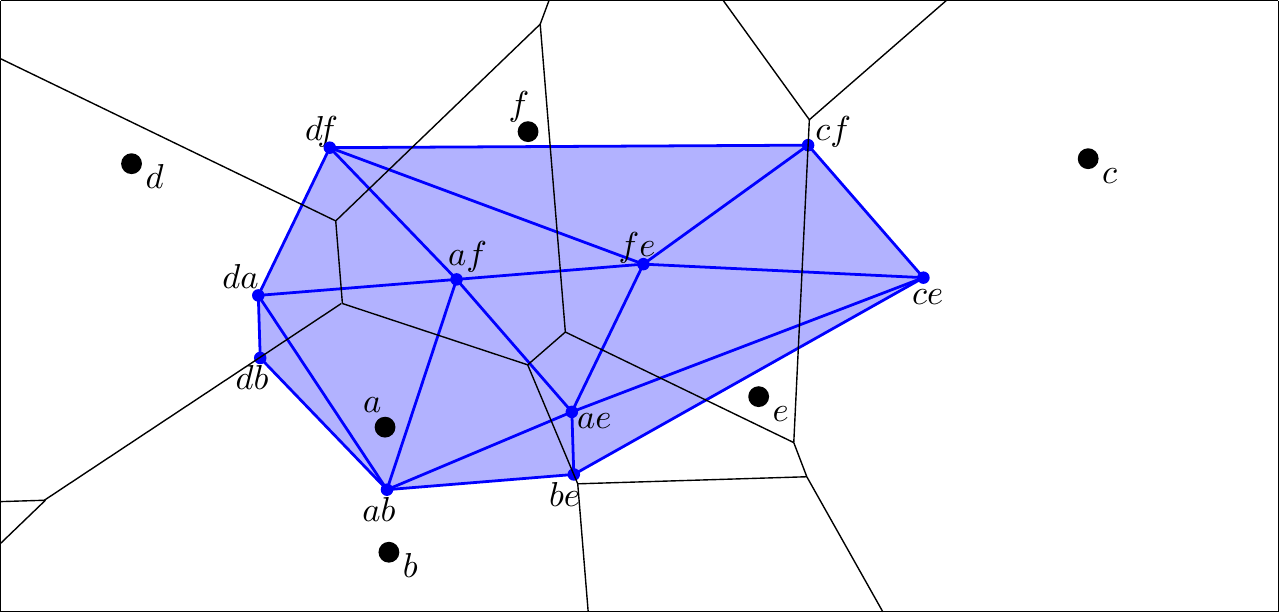}
  \caption{Superposition of the order-$2$ Voronoi tessellation (in \emph{black})
    and the order-$2$ Delaunay mosaic (in \emph{blue})
    of a set of six points in the plane.
    Each domain of the tessellation corresponds to two of these six points,
    and the corresponding vertex of the mosaic is the average of these two points.}
  \label{fig:order2_full}
\end{figure}

Assuming $A$ is in general position, \cite{EdOs18} established the existence
of a tiling in $\Rspace^{d+1}$ whose horizontal integer slices are the
Delaunay mosaics.
We recall the definition of the tiling and its most important properties.
Let $A \subseteq \Rspace^d$ be locally finite and in general position.
We construct a rhomboid, $\rho$,
for each partition $A = \Ain \sqcup \Aon \sqcup \Aout$
for which there exists a sphere $S$ in $\Rspace^d$ such that
all points in $\Ain$ lie inside $S$,
all points in $\Aon$ lie on $S$,
and all points in $\Aout$ lie outside $S$.
Whenever convenient, we write $\Ain(\rho) = \Ain$, $\Aon(\rho) = \Aon$,
and $\Aout(\rho) = \Aout$ to indicate the correspondence.
Due to general position of $A$, we have $\card{\Aon} \leq d+1$.
A \emph{combinatorial vertex} of $\rho$ is a collection of points
that contains $\Ain$ and is contained in $\Ain \cup \Aon$,
and we write
\begin{align}
  \vx{\rho}  &=  \{ \Ain \subseteq Q \subseteq \Ain \cup \Aon \}
  \label{eqn:newvx}
\end{align}
for the collection of combinatorial vertices of $\rho$.
Setting $y_a = (a, -1) \in \Rspace^{d+1}$, for every $a \in A$,
we draw the rhomboids in $\Rspace^{d+1}$ by mapping every combinatorial vertex
to $y_Q = \sum_{q \in Q} y_q$, in which $y_\emptyset = 0$, by convention.
The $(d+1)$-st coordinate of $y_Q$ is therefore $- \card{Q}$,
and we call $\card{Q}$ the \emph{depth} of the vertex.
The geometric realization of a rhomboid $\rho$ is the convex hull
of the locations of its combinatorial vertices, which is a rhomboid.
We refer to $\Ain(\rho)$ as the \emph{anchor vertex} of $\rho$.

The \emph{rhomboid tiling} of $A$, denoted $\Rhomboid{A}$, is the collection
of thus defined rhomboids.
By assumption of general position, every face of a rhomboid is again defined
by a sphere as described above and thus belongs to the rhomboid tiling.
As proved in \cite{EdOs18}, any two rhomboids are either disjoint or intersect
in a common face, which implies that the rhomboid tiling is a complex
embedded in $\Rspace^{d+1}$; see Figure \ref{fig:duality} for an example.
The following properties have been observed in \cite{EdOs18}:
\begin{proposition}[Rhomboid Tiling]
  \label{prop:RhomboidTiling}
  Let $A \subseteq \Rspace^d$ be locally finite and in general position.
  \begin{enumerate}
    \item[1.] $\Rhomboid{A}$ is dual to an arrangement of hyperplanes in $\Rspace^{d+1}$;
    \item[2.] the slice of $\Rhomboid{A}$ at depth $k$
              is the order-$k$ Delaunay mosaic of $A$, scaled by a factor $k$.
  \end{enumerate}
\end{proposition}
\medskip
The hyperplane arrangement will be introduced in Section \ref{sec:vertexset}.
We elaborate on the second property: that each cell of the order-$k$ Delaunay mosaic
is a slice of some rhomboid.
Combinatorially, each rhomboid is a cube and, again combinatorially,
each cell of $\Del{k}{A}$ is a slice
orthogonal to the cube diagonal that passes through a non-empty set of the vertices.
For the $(d+1)$-cube, there are $d+2$ such slices, which we index from top to bottom
by the \emph{generation} $0 \leq g \leq d+1$.
The $g$-th slice passes through $\binom{d+1}{g}$ vertices,
so we have a vertex at generations $g = 0,d+1$,
a $d$-simplex at generations $g = 1,d$,
and some other $d$-dimensional polytope at generations $2 \leq g \leq d-1$.
In $d+1 = 3$ dimensions, we have a vertex, a triangle, another triangle, and another vertex,
see Figure \ref{fig:full_3d};
but already in $d+1 = 4$ dimensions, the middle slice is not a simplex;
see Figure \ref{fig:slices}.
We remark that in addition to the order-$k$ Delaunay mosaic,
also the degree-$k$ Delaunay mosaic,
which is the dual of the degree-$k$ Voronoi tessellation \cite{EdSe86},
can be obtained as a slice of $\Rhomboid{A}$ at depth $k - \frac{1}{2}$.
\begin{figure}[hbt]
  \centering
  \includegraphics[width=0.4\textwidth, page=2]{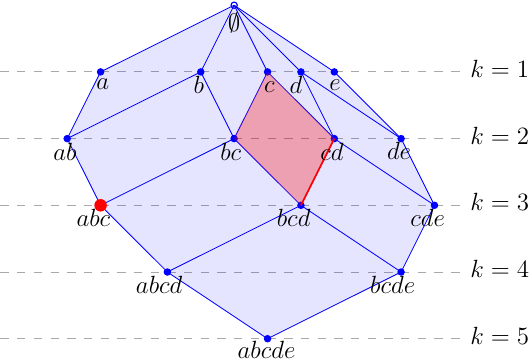}
  \hspace{0.3in}
  \includegraphics[width=0.4\textwidth, page=2]{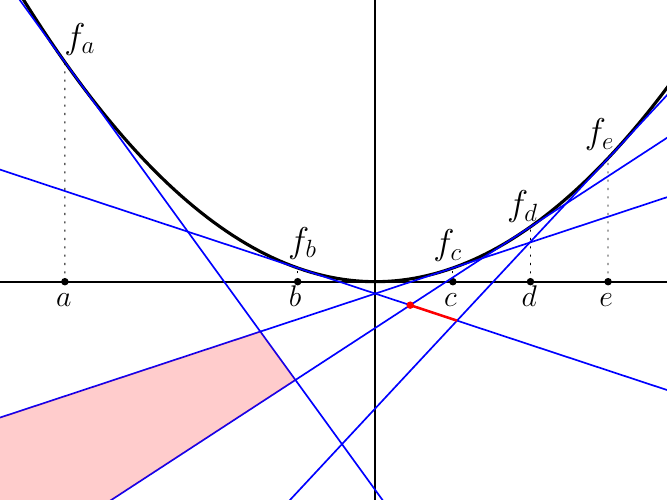}
  \vspace{0.1in}
  \caption{\emph{Left:} the rhomboid tiling of five points in $\Rspace^1$.
    The highlighted rhomboid defined by $\Ain = \{c\}$ and $\Aon = \{b,d\}$
    is the convex hull of the points $y_c$, $y_{\{b,c\}}$, $y_{\{c,d\}}$,
    and $y_{\{b,c,d\}}$.
    The horizontal line at depth $k$ intersects the tiling in
    the order-$k$ Delaunay mosaic.
    \emph{Right:} the dual hyperplane arrangement.
    Following the dotted lines connecting the points of $A$ on the horizontal axis
    to the paraboloid, we find the corresponding tangent hyperplanes.
    The highlighted rhomboids of dimension $j = 0, 1, 2$ are dual to the
    highlighted cells of dimension $2-j$ in the arrangement.}
  \label{fig:duality}
\end{figure}
\begin{figure}[hbt]
  \centering
  \vspace{0.1in}
  \includegraphics[width=0.90\textwidth]{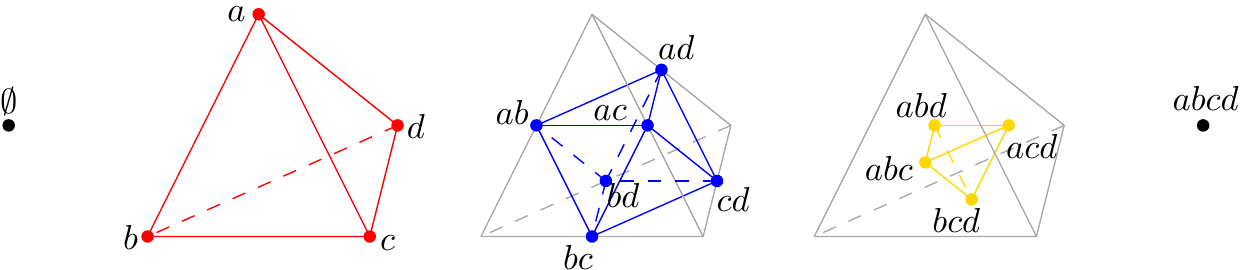}
  \caption{Slices of a $4$-dimensional rhomboid defined by
    $\Ain(\rho) = \emptyset$ and $\Aon(\rho) = \{a,b,c,d\}$.
    The non-trivial slices are a tetrahedron at generation $g=1$,
    an octahedron at generation $g=2$, and another tetrahedron at generation $g=3$.}
  \label{fig:slices}  
\end{figure}

\section{Combinatorial Properties}
\label{sec:combinatorial}

As proved in \cite{aurenhammer1990}, the order-$k$ Delaunay mosaic
is the projection of the boundary complex of a convex polyhedron in $\Rspace^{d+1}$.
To explain this construction, we define the \emph{lift} of $a \in \Rspace^d$
as the point $\lift{a} = (a, \norm{a}^2) \in \Rspace^{d+1}$.
For each $k$-tuple $Q \subseteq A$,
we take the barycenter of their lifts, $\frac{1}{k}\sum_{q\in Q} \lift{q}$,
and obtain the order-$k$ Delaunay mosaic as the vertical projection
of the lower faces of the convex hull of these barycenters.
Equivalently, we can interpret each barycenter of lifts as a weighted point
in $\Rspace^d$ and get the order-$k$ Delaunay mosaic as the weighted
order-$1$ Delaunay mosaic of the weighted points.
By itself, this approach does not scale well with $k$
since there are $\binom{\card{A}}{k}$ such barycenters.
Most barycenters, however, are irrelevant as they do not contribute
to the lower faces of the convex hull.
If we could, somehow, identify the relevant barycenters without wasting time
on the irrelevant ones,
this procedure would efficiently construct the cells of the order-$k$ Delaunay mosaic
by computing the weighted first-order Delaunay mosaic.
We will see how this can be done in Section \ref{sec:vertexset}.

In $d \geq 3$ dimension, not all cells of $\Del{k}{A}$ are simplicial,
even if the points in $A$ are in general position.
The cells carry important information, which for some applications is essential
and cannot be easily recovered from a triangulation.
This poses an additional challenge because most algorithms for computing
convex hulls or weighted first-order Delaunay mosaics return
a triangulated version of the correct mosaic.
As explained in the following section, 
we address this issue by predicting the cells from their corresponding rhomboids.

\subsection{Predicting Cells}
Given a cell $\sigma$ in the order-$k$ Delaunay mosaic,
the following lemma identifies the rhomboid, $\rho$, that $\sigma$ is a slice of;
see Figure \ref{fig:full_3d} for an illustration.
Write $\vx{\sigma}$ for the set of combinatorial vertices whose locations
are the vertices of $\sigma$.
Clearly, $\vx{\sigma} \subseteq \vx{\rho}$.
\begin{lemma}
  \label{lem:xinxon}
  Let $\rho \in \Rhomboid{A}$ be a rhomboid
  and $\sigma \in \Del{k}{A}$ a slice of $\rho$.
  Then $\Ain(\rho) = \bigcap \vx{\sigma}$,
  $\Aon(\rho) = \bigcup \vx{\sigma} \setminus \Ain(\rho)$,
  and the generation of $\sigma$ is $k - \card{\Ain(\rho)}$.
\end{lemma}
\begin{proof}
  Recall that
  $\vx{\rho} = \{\Ain(\rho) \subseteq Q \subseteq \Ain(\rho) \cup \Aon(\rho) \}$,
  in which $\Ain(\rho)$ and $\Aon(\rho)$ are disjoint.
  Since the depth of a vertex is determined by its cardinality, and the vertices
  of a slice are by definition all at the same depth, the vertices
  of the generation-$g$ slice all satisfy $\card{Q} - \card{\Ain(\rho)} = g$.
  The intersection of all $g$-subsets of $\Aon(\rho)$ is of course empty,
  which implies that the intersection of the combinatorial vertices
  of the slice is $\Ain(\rho)$.
  Furthermore, $\bigcup \vx{\sigma} = \bigcup \vx{\rho}$ for every slice
  $\sigma$ of $\rho$ with generation $g \geq 1$.
  The union of all $g$-subsets of $\Aon(\rho)$ is $\Aon(\rho)$
  itself, and thus $\Aon(\rho) = \bigcup \vx{\sigma} \setminus \Ain(\rho)$.
  Finally, the generation of $\sigma$ is the difference in depth of
  the anchor vertex, $\Ain(\rho)$, and the slice defining $\sigma$. 
  The depth of $\sigma$ is $k$ and the depth of 
  $\Ain(\rho)$ is its cardinality, which completes the proof.
\end{proof}

\begin{figure}[hbt]
  \centering
  \includegraphics[width=0.70\textwidth, page=18]{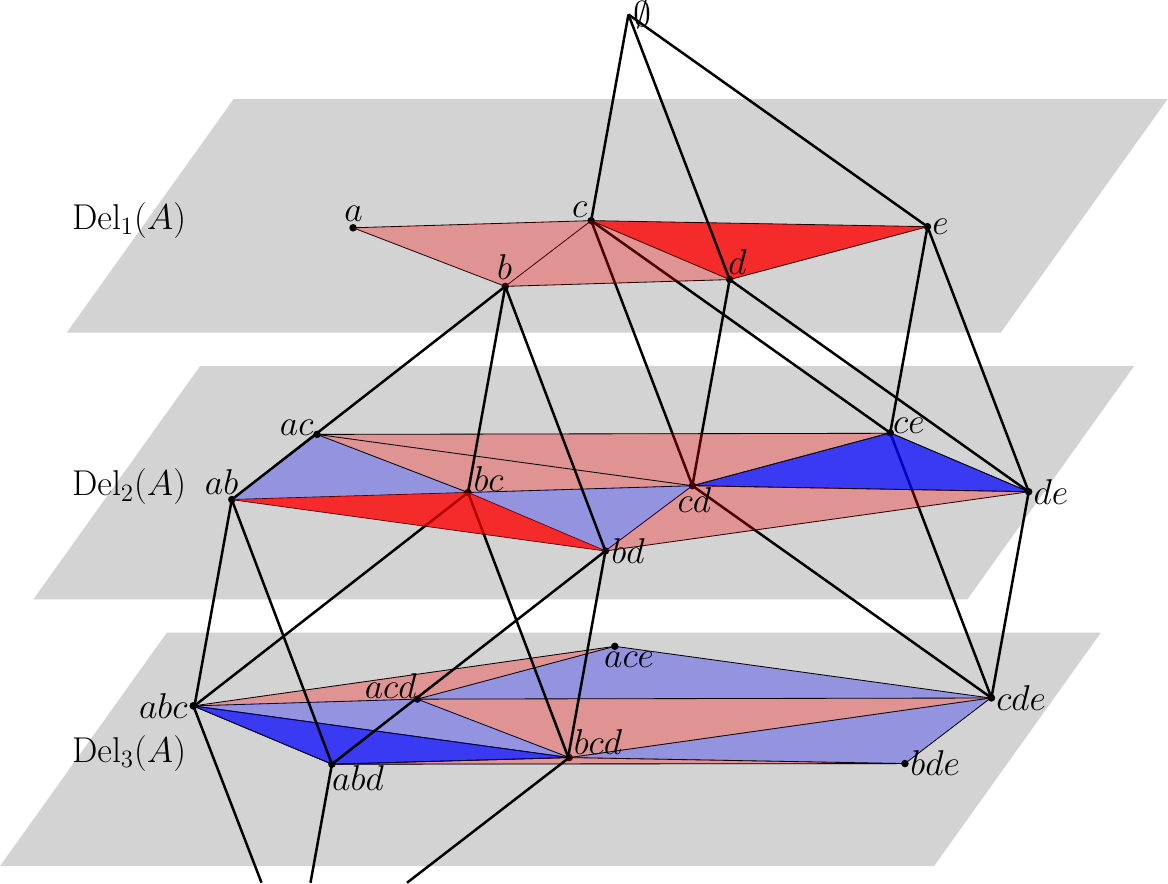}
  \caption{First-, second-, and third-order Delaunay mosaics of the
    set $A = \{a,b,c,d,e\}$ in $\Rspace^2$ as slices of the $3$-dimensional
    rhomboid tiling.
    For clarity, only two of the rhomboids are shown, with their
    first-generation slices in \emph{red} and second-generation slices
    in \emph{dark blue}.
    The rhomboids on the \emph{left} and \emph{right} are defined by
    $\Ain = \{b\}, \Aon = \{a,c,d\}$ and
    $\Ain = \emptyset, \Aon = \{c,d,e\}$, respectively.}
  \label{fig:full_3d}
\end{figure}

If all of our order-$k$ Delaunay cells are triangulated---e.g.\ due to being
the output of a weighted first-order Delaunay algorithm---we cannot directly
apply Lemma \ref{lem:xinxon}.
Indeed, if $\tau$ is a simplex that is part of a triangulation
of a slice $\sigma$ of a rhomboid $\rho$, then
$\bigcap \vx{\tau}$ and $\bigcup \vx{\tau}$ do not necessarily
equal $\Ain (\rho)$ and $\Ain (\rho) \cup \Aon (\rho)$.
We can, however, still identify whether $\tau$ is a first-generation slice
of $\rho$ and thus in fact is equal to $\sigma$.
Using Lemma \ref{lem:xinxon}, we can then obtain $\rho$.
\begin{lemma}
  \label{lem:identify_generation}
  A $d$-simplex, $\tau$, in a triangulation of $\Del{k}{A}$ is
  a first-generation $d$-cell of $\Del{k}{A}$ if and only if the 
  intersection of its combinatorial vertices is of size $k-1$.
\end{lemma}
\begin{proof}
  Let $\sigma$ be the $d$-cell in $\Del{k}{A}$ that contains $\tau$
  in its triangulation, and assume $\sigma$ is a generation-$g$ slice of $\rho$.
  From Lemma \ref{lem:xinxon}, we know that $\Ain(\rho) \subseteq v$ for all
  $v \in \vx{\sigma}$, and $\card{\Ain(\sigma)} = k-g$.
  The remaining $g$ points in every $v$ are from $\Aon(\rho)$.
  We have $\vx{\tau} \subseteq \vx{\sigma}$ with $\card{\vx{\tau}} = d+1$.
  So for $\tau$ to consist of vertices whose common 
  intersection is of size $k-1$, there need to be $d+1$ distinct $g$-subsets
  of $\Aon(\rho)$ that all have $g-1$ points in common. 
  However, as $\card{\Aon(\rho)} = d+1$, this is not possible unless $g=1$.
\end{proof}

\subsection{Identifying Vertices}
\label{sec:vertexset}

Given a triangulation of the order-$k$ Delaunay mosaic,
we just saw how to identify its first-generation cells.
From these, we can obtain the corresponding rhomboids
and their higher-generation slices.
We shall now prove that if we have triangulations of the order-$j$ Delaunay
mosaics, for all $j < k$,
we can assemble the complete vertex set of the order-$k$ Delaunay mosaic
by taking slices at depth $k$ obtained
from first-generation cells at lower depths.
We note that this only holds in the unweighted setting.

\medskip
To prepare the proof of this result, we recall the definition of
the hyperplane arrangement postulated by Proposition \ref{prop:RhomboidTiling}.
For each point $a \in A$, write $f_a \colon \Rspace^d \to \Rspace$
for the affine map defined by
$f_a (x) = 2 \scalprod{x}{a} - \norm{a}^2 = \norm{x}^2 - \Edist{x}{a}^2$.
The graph of $f_a$ is a hyperplane in $\Rspace^{d+1}$ that is tangent to the
paraboloid $\Paraboloid$ of points $(x,z) \in \Rspace^d \times \Rspace$ that satisfy
$z = \norm{x}^2$.
The collection of such hyperplanes decomposes $\Rspace^{d+1}$ into convex cells,
which we call the \emph{hyperplane arrangement} of $A$, denoted $\Arr{A}$;
see Figure \ref{fig:duality}.
The \emph{cells} in the arrangement are intersections of hyperplanes
and closed half-spaces.
More formally, for each cell there is
an ordered three-partition $A = \Ain \sqcup \Aon \sqcup \Aout$
such that the cell consists of all points $(x, z) \in \Rspace^d \times \Rspace$
that satisfy $z \leq f_a (x)$, if $a \in \Ain$;
$z = f_a (x)$, if $a \in \Aon$; and
$z \geq f_a (x)$, if $a \in \Aout$.
This three-partition is the key to establishing the bijection between
the cells of $\Arr{A}$ and the rhomboids of $\Rhomboid{A}$
that proves the duality claimed in Proposition \ref{prop:RhomboidTiling}.
We call top-dimensional cells of $\Arr{A}$ \emph{chambers};
they satisfy $\Aon = \emptyset$.
The depth of a chamber is $\card{\Ain}$ or, equivalently, the number of hyperplanes
that are above this chamber; it equals the depth of the dual vertex in $\Rhomboid{A}$.
To see the aforementioned relationship between the arrangement
and the higher-order Voronoi tessellations, we observe that
the chamber in $\Arr{A}$ with $\Ain = Q$ vertically projects to $\domain{Q}$.
We can therefore construct $\Vor{k}{A}$ by computing and projecting all chambers
whose ordered three-partitions satisfy $\card{\Ain} = k$;
see \cite[Chapter 13]{Ede87} or \cite{EdSe86}.

\medskip
We call a chamber $\gamma$ a \emph{bowl} if only one of its facets 
bounds it from above or,
equivalently, if there is only one chamber at the next lower depth
that shares a facet with $\gamma$.
We call the hyperplane that contains this facet the \emph{lid} of the bowl.
\begin{lemma} \label{lem:pbad}
  A hyperplane that is a lid of a bowl at depth $1$ is not a lid of any other chambers.
\end{lemma}
\begin{proof}
  Let $\gamma$ be a bowl at arbitrary depth, and let $P$ be its lid.
  Every other hyperplane that contains a facet of $\gamma$
  bounds $\gamma$ from below.
  The top facet of $\gamma$ is the only part of $P$ that is above all of these hyperplanes;
  that is: all other parts of $P$ are below at least one of the other hyperplanes. 
  This implies that every other bowl with lid $P$ has at least one other hyperplane above it,
  and is thus of depth at least $2$.

  Now assume $\gamma$ is at depth $1$.
  If there were another bowl $\gamma'$ with lid $P$, 
  then the above argument would yield that all other bowls are at depth
  at least $2$, contradicting our assumption on $\gamma$.
  Thus $\gamma$ has to be the unique bowl with lid $P$.
\end{proof}

With this lemma, we are ready to state and prove the main combinatorial insight
that motivates our algorithm.
In a nutshell, it says that the first-generation cells form clusters
without interior vertices.
In $\Rspace^2$, this is equivalent to saying that these clusters
have outer-planar $1$-skeletons.
\begin{theorem} \label{thm:orderk-vertices}
  Let $A \subseteq \Rspace^d$ be locally finite and $k \geq 2$.
  Then every vertex in $\Del{k}{A}$ is vertex of some $d$-cell
  of generation $g \geq 2$.
\end{theorem}
\begin{proof}
  In the unweighted setting, each hyperplane is tangent to the paraboloid $\Paraboloid$ and
  contains a facet of the unique depth-$0$ chamber.
  Thus, each hyperplane is the lid to a chamber at depth $1$.
  As this is true for every hyperplane,
  all chambers of depth $2$ or higher have no lids by Lemma \ref{lem:pbad}.
  This means that any chamber of depth at least $2$ has at least two upper facets.
  Because the upper boundary is connected, there are two upper facets that meet
  in a $(d-1)$-face, the dual rhomboid of this face has dimension $2$,
  and its bottom vertex is dual to the chamber.
  Thus we can obtain this vertex, $v$, knowing the other three vertices
  of the $2$-rhomboid.

  Any $2$-rhomboid is a face of some $(d+1)$-dimensional rhomboid, $\rho$, 
  which thus contains $v$ at generation at least $2$,
  i.e.\ $v$ has depth at least $\card{\Ain(\rho)} + 2$.
  Knowing $\Ain(\rho)$ and $\Aon(\rho)$,
  we obtain this vertex via Equation \eqref{eqn:newvx}.
\end{proof}

\section{Algorithm}
\label{sec:algorithm}

We outline our algorithm in this section;
its correctness follows from the results of the previous sections.
We compute the Delaunay mosaics one by one in sequence of increasing order.
For $\Del{1}{A}$, the vertex set is the set $A$ of input points.
Whenever we have the vertex set of $\Del{j}{A}$, we compute
its (triangulated) $d$-cells using an off-the-shelf algorithm for
weighted Delaunay triangulations. We use Lemma \ref{lem:identify_generation} to
identify the first-generation $d$-cells, while discarding all other cells.
From each first-generation cell, we obtain and save
the higher-generation $d$-cells and vertices defined by the same rhomboid.
These will appear in Delaunay mosaics of higher orders.
By Theorem \ref{thm:orderk-vertices}, once we have processed all
$\Del{j}{A}$ for $j < k$, we will have obtained the complete vertex set
of $\Del{k}{A}$ in the process, thus allowing our algorithm to
continue until we have all Delaunay mosaics up to the desired order.

Algorithm \ref{alg:okdel2} is a more formal write-up of the above outline,
and Figure \ref{fig:full_3d} visualizes the process.
A dimension-agnostic {\tt python} implementation and
a 2- and 3-dimensional {\tt C++} implementation using
CGAL \cite{cgal:eb-19b} are available at \cite{orderkgithub,rhomboidgithub}.
If we store all first-generation cells with their anchor vertices,
we can use this algorithm to implicitly construct the rhomboid tiling,
as done in \cite{rhomboidgithub}.

\begin{algorithm}
  \caption{computes the order-$k$ Delaunay mosaic of a finite set
    of (unweighted) points, $A \subseteq \Rspace^d$.
    We represent each $d$-cell of $\Del{k}{A}$ by the collection of 
    its combinatorial vertices,
    stressing that these collections are sets and thus contain every
    combinatorial vertex only once.
    Duplicity is avoided by checking before adding.
    The locations of the combinatorial vertices are the barycenters of their points,
    and the cell is the convex hull of these locations.

    While the software for computing the weighted Delaunay mosaic may return
    all cells triangulated, our algorithm outputs the (non-triangulated) cells
    of the order-$k$ Delaunay mosaic.
    We recall that in $d \geq 3$ dimensions such non-simplicial cells
    appear generically for $k \geq 2$.}
  \label{alg:okdel2}
  \begin{algorithmic}
  \nolinenumbers
    \State $\vx{\Del{1}{A}} := A$
    \For{$j$ {\bf from} $1$ {\bf to} $k$}
      \State\Comment{Compute the location and weight of each combinatorial vertex}
      \ForAll{$v \in \vx{\Del{j}{A}}$}
        \State $loc(v) := \frac{1}{j} \sum_{a \in v} a$
        \State $wt(v)  := \norm{\frac{1}{j}\sum_{a\in v} a}^2 - \frac{1}{j}\sum_{a \in v} \norm{a}^2$
      \EndFor
      \State\Comment{Get the (triangulated) cells of the order-$j$ Delaunay mosaic}
      \State $D$ := \texttt{weightedDelaunay}($loc$, $wt$)
      \State\Comment{Infer vertices and higher-generation cells of later Delaunay mosaics}
      \ForAll{$d$-simplices $\sigma$ in $D$}
        \State\Comment{Check whether the generation of $\sigma$ is 1 via Lemma \ref{lem:identify_generation}}
        \State\Comment{We already obtained higher-generation cells of $\Del{j}{A}$ earlier.}
        \If{$\card{\bigcap \vx{\sigma}} = j-1$} 
          \State Add $\sigma$ to $\Del{j}{A}$
          \State\Comment{Get $\Ain(\rho)$ and $\Aon(\rho)$ via Lemma \ref{lem:xinxon}}
          \State $\Ain(\rho) := \bigcap \vx{\sigma}$
          \State $\Aon(\rho) := \bigcup \vx{\sigma} \setminus \Ain(\rho)$
          \For{$g$ {\bf from} $2$ {\bf to} $d$}
            \State\Comment{Get the generation-$g$ cell, $\sigma'$, of the rhomboid of $\sigma$, via Equation \eqref{eqn:newvx}}
            \State $\vx{\sigma'} := \{\Ain(\rho) \cup Q \mid Q \in \Aon(\rho), \card{Q} = g\}$
            \State Add all $v \in \vx{\sigma'}$ to $\vx{\Del{j+g-1}{A}}$
            \State Add $\sigma'$ to $\Del{j+g-1}{A}$
          \EndFor
        \EndIf
      \EndFor
    \EndFor
    \State\Return $\vx{\Del{k}{A}}$, $\Del{k}{A}$
  \end{algorithmic}
\end{algorithm}

\medskip
To get a handle on the runtime of the algorithm, we consider the two steps
used to compute the order-$k$ Delaunay mosaic after finishing the construction
of the first $k-1$ mosaics.
The first step is geometric and invokes the black-box algorithm
to construct the weighted Delaunay mosaic from which we get vertices and cells
of (unweighted) higher-order Delaunay mosaics.
The runtime of this step depends on the runtime of the black box algorithm,
which in many cases is output-dependent.
The second step is combinatorial and determines, for each output simplex
from the first step, whether it is first generation,
in which case it is a genuine cell of the mosaic.
Assuming constant dimension, $d$, identifying whether an order-$(k-1)$ cell
is of first generation and, in this case, obtaining the higher-generation cells
takes time $O(k)$.
Thus, for a given $k$, the combinatorial step takes time $O(k C_k)$,
in which $C_k$ is the number of $d$-cells of $\Del{k}{A}$.
With each vertex being represented as a $k$-tuple of points,
this is linear in the output size,
assuming we store each cell naively as a set of its vertices. 
If the runtime of each black-box invocation were linear in the output size,
the total runtime for producing the first $k$ higher-order Delaunay mosaics
would thus be linear in the output size as well.

In practice, it is more efficient to store a cell as a set
of pointers or indices to its vertices, only requiring space $O(kV_k + C_k)$,
with $V_k$ denoting the number of vertices of $\Del{k}{A}$.
Using this representation, the combinatorial step is not linear in the output
size unless the number of cells of $\Del{k}{A}$ is linear
in the number of vertices.

\section{Experimental Results}
\label{sec:experimental}  

In $2$ dimensions, the number of cells in the (order-$1$) Delaunay mosaic
is always linear in the number of input points,
while in $d \geq 3$ dimensions, the size of the mosaic depends on the input set
itself---and not just its cardinality---and ranges from
$\Omega(n)$ to $O(n^{\ceil{d/2}})$ \cite{mcmullen1970maximum}.
The asymptotic worst case is realized by points located on
the moment curve, $(t, t^2, \dots, t^d)$ with $t \in \Rspace$,
while e.g.\ uniformly sampled points within a sphere have expected linear
size \cite{dwyer1991higher}, as do uniformly sampled points on a
convex polytope in $\Rspace^3$ \cite{golin2003average}.
Under appropriate sampling conditions for points on a smooth surface,
the size of the mosaic is $O(n \log n)$ \cite{attali2003complexity}.

\medskip \noindent
\emph{Size in $3$ dimensions.}
To shed light on the size range of order-$k$ Delaunay mosaics,
we compute them for a few $3$-dimensional point sets relevant to these bounds.
Note that for order-$k$ Delaunay mosaics the number of vertices varies as well.
Figure \ref{fig:comparison} shows the numbers of vertices and $3$-dimensional cells
for all higher-order Delaunay mosaics of four sets of size $n=200$ each:
points on the {\tt moment curve},
points sampled on the {\tt torus} (with major radius $1$ and minor radius $0.5$
obtained by uniformly sampling the angles of its parametrization), 
points uniformly sampled inside the {\tt unit ball},
and a point set in convex position forming a {\tt polytope}
(obtained by uniformly sampling points inside
a ball and randomly choosing $200$ vertices of the convex hull).
\begin{figure}[htb]
  \centering
  \begin{subfigure}[t]{0.48\textwidth}
    \centering
    \includegraphics[width=\textwidth]{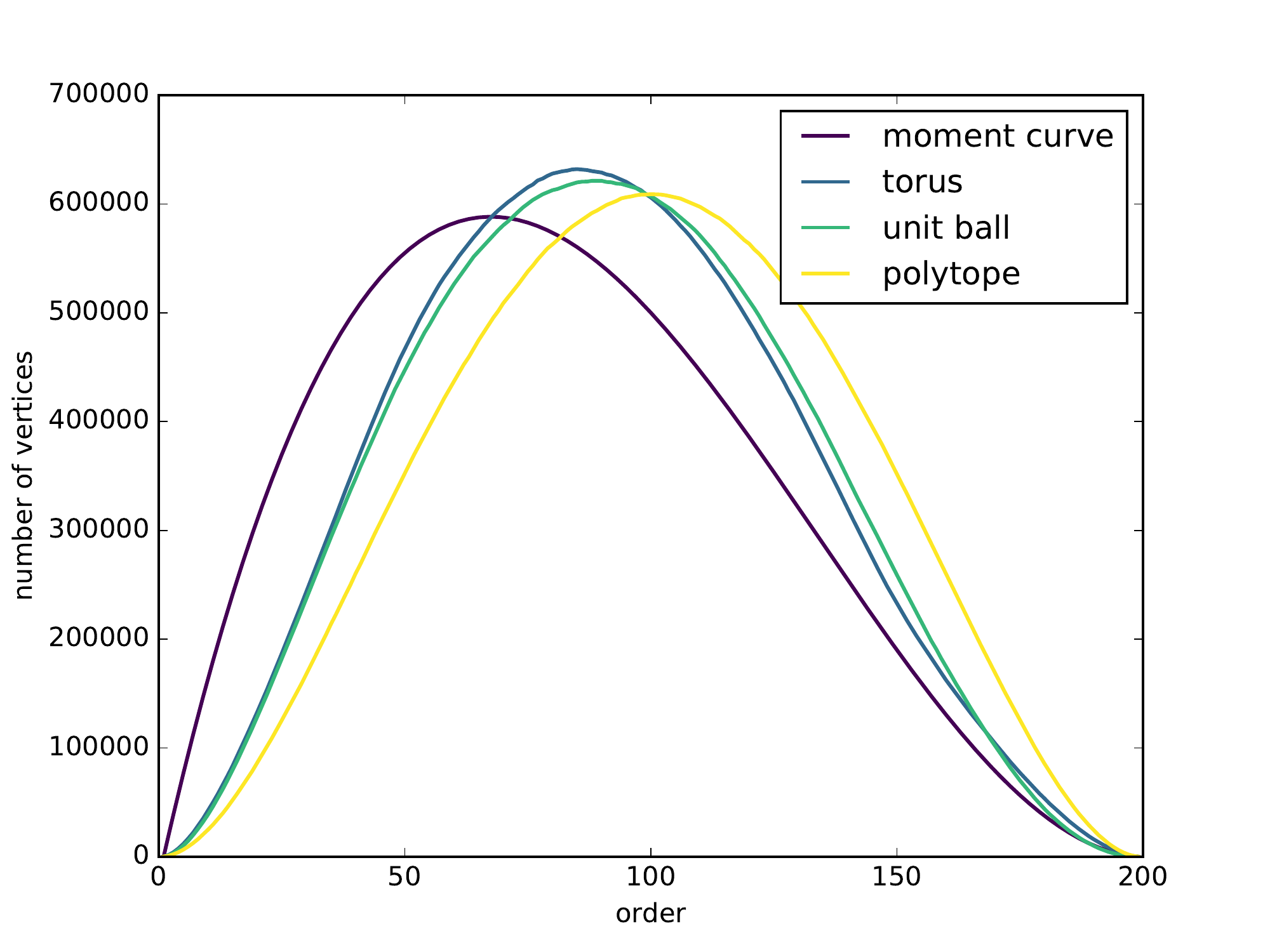}
  \end{subfigure}
  ~ 
  \begin{subfigure}[t]{0.48\textwidth}
    \centering
    \includegraphics[width=\textwidth]{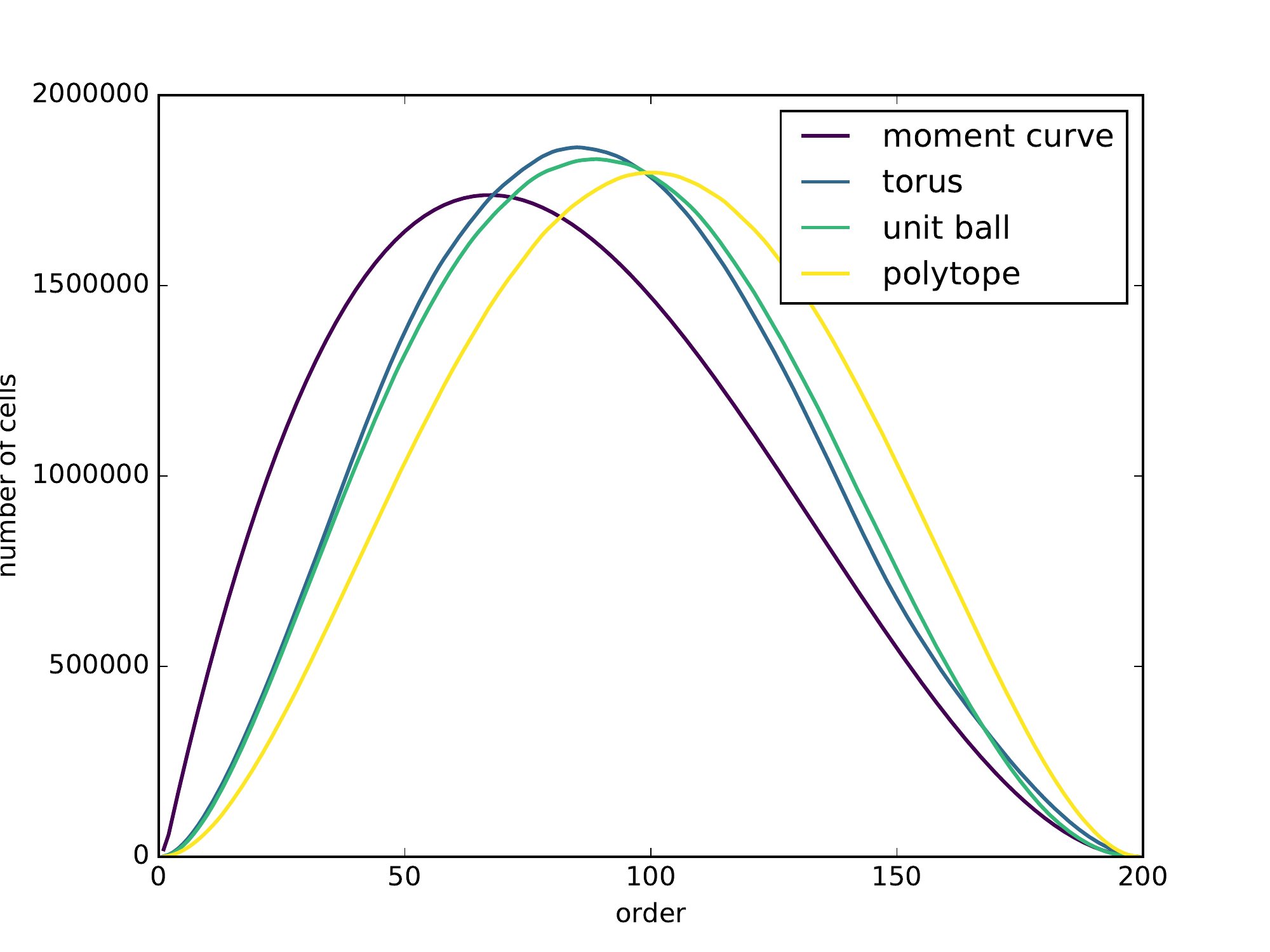}
  \end{subfigure}
  \caption{Number of vertices (\emph{left})
    and of $3$-dimensional cells (\emph{right})
    in the order-$k$ Delaunay mosaics of four sets with $n = 200$
    points in $\Rspace^3$ each.}
  \label{fig:comparison}  
\end{figure}
The plots of vertex numbers and cell numbers generally
resemble each other, with roughly three times as many cells as vertices.
Other than in Figure \ref{fig:comparison}, we therefore omit the information
about the vertices and show only the plots for the cells.
The {\tt moment curve} and {\tt polytope} sets are both in convex position.
Nevertheless, the size of the mosaic for the moment curve grows large
faster for small $k$,
and reaches its peak at $k \approx n/3$,
while for the polytope the peak is at $k \approx n/2$.
Notice how a faster rise also goes along with an earlier decay.
This is a consequence of the fact that the total size of all
order-$k$ Delaunay mosaics together---or, equivalently of the rhomboid tiling---only
depends on the input size, $n$, and not on the relative position of the input points.

  \medskip \noindent
  \emph{Size increase for small order.}
Looking more closely at the growth for small $k$ relative to the input size,
we observe that
the {\tt polytope} and {\tt unit ball} exhibit linear growth
while the size of the mosaic seems to grow quadratically
for the {\tt moment curve}, see Figure \ref{fig:complexity}.
This is consistent with the bounds on first-order Delaunay mosaics mentioned earlier.
For the {\tt torus}, the size seems to grow slightly superlinearly,
which is again consistent with the $O(n \log n)$ bound for smooth surfaces
mentioned above.

\begin{figure}[htb]
  \centering
  \begin{subfigure}[t]{0.48\textwidth}
    \centering
    \includegraphics[width=\textwidth]{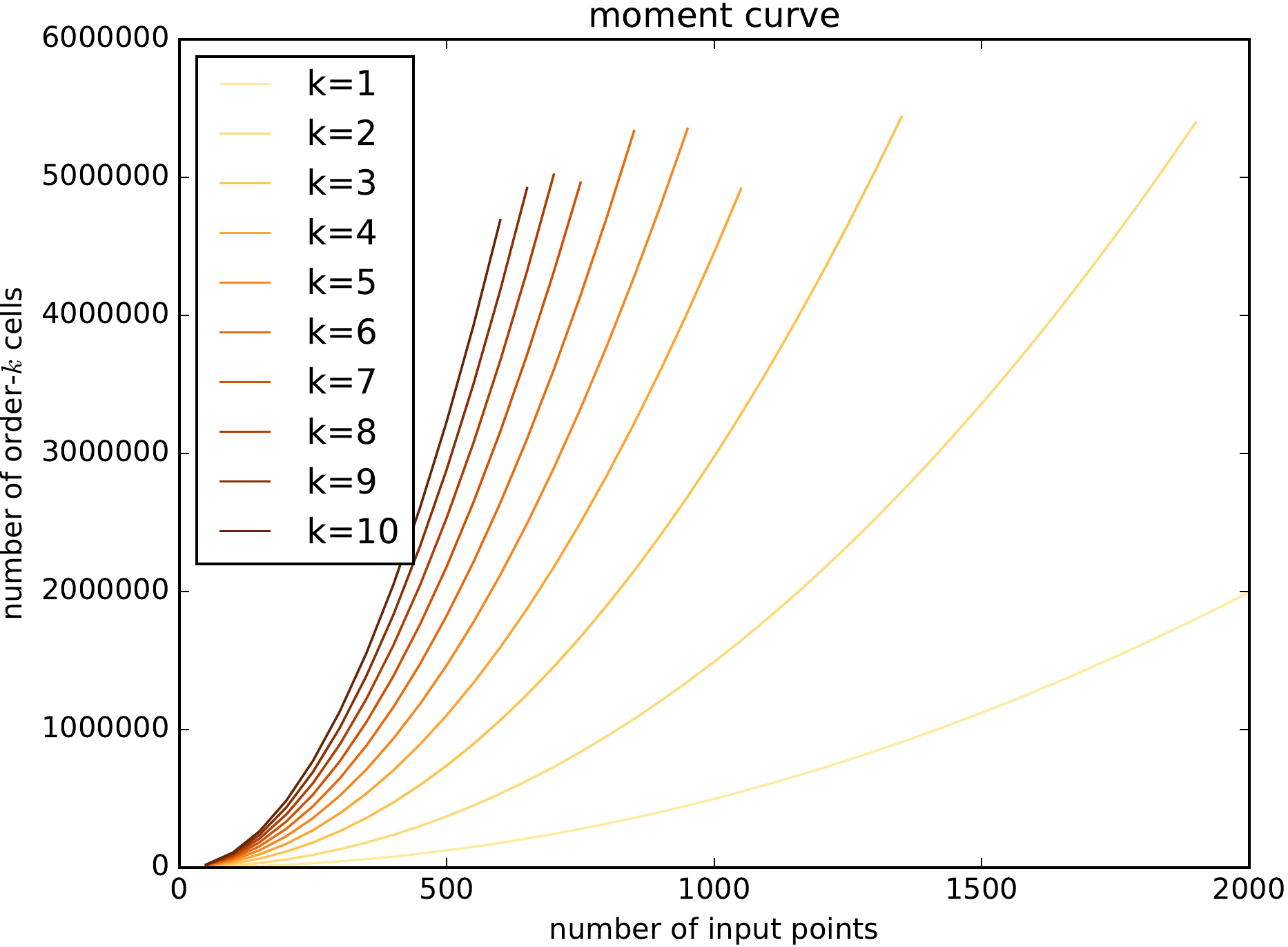}
  \end{subfigure}
  ~
  \begin{subfigure}[t]{0.48\textwidth}
    \centering
    \includegraphics[width=\textwidth]{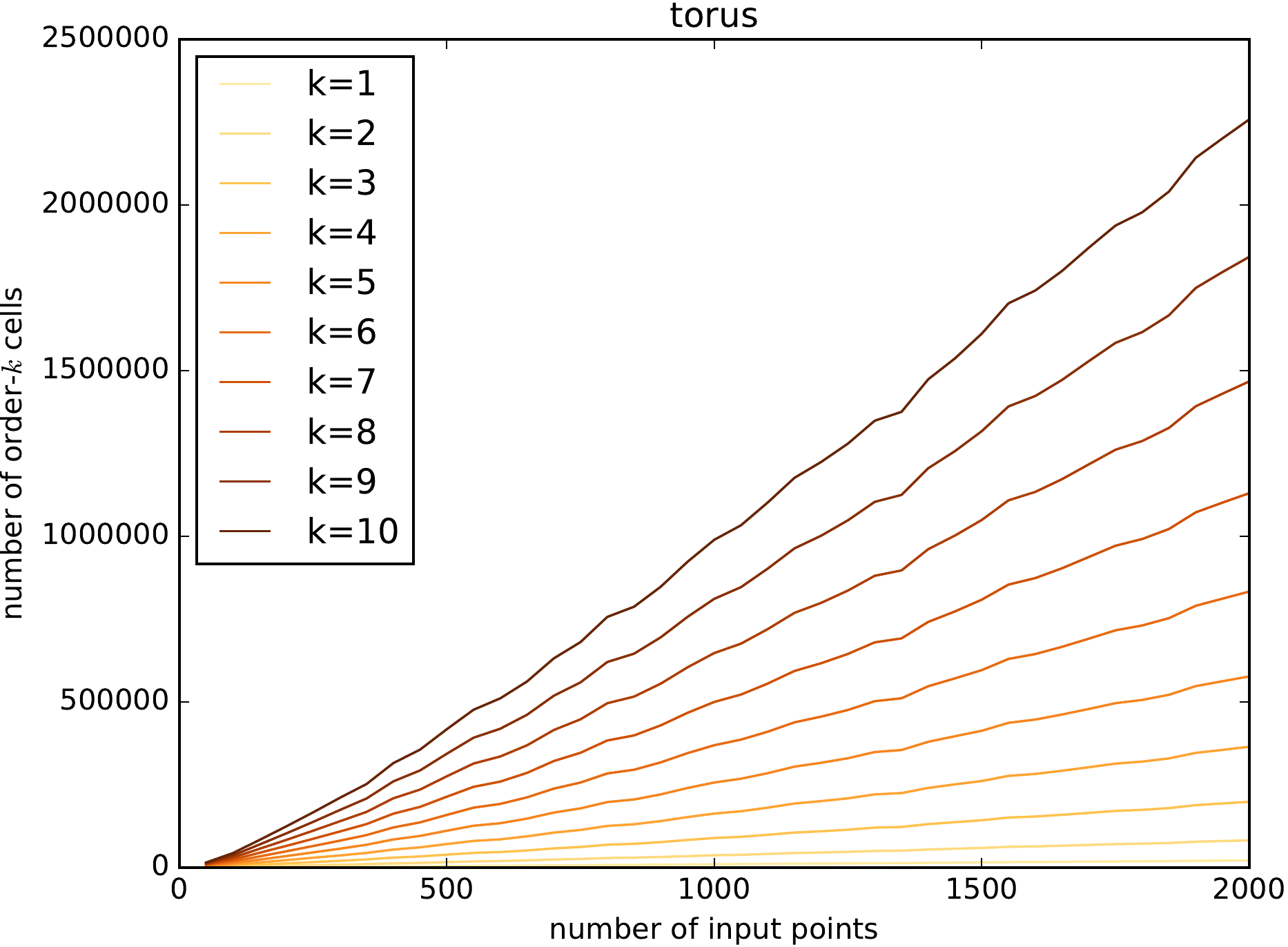}
  \end{subfigure}
  ~ 
  \begin{subfigure}[t]{0.48\textwidth}
      \centering
      \includegraphics[width=\textwidth]{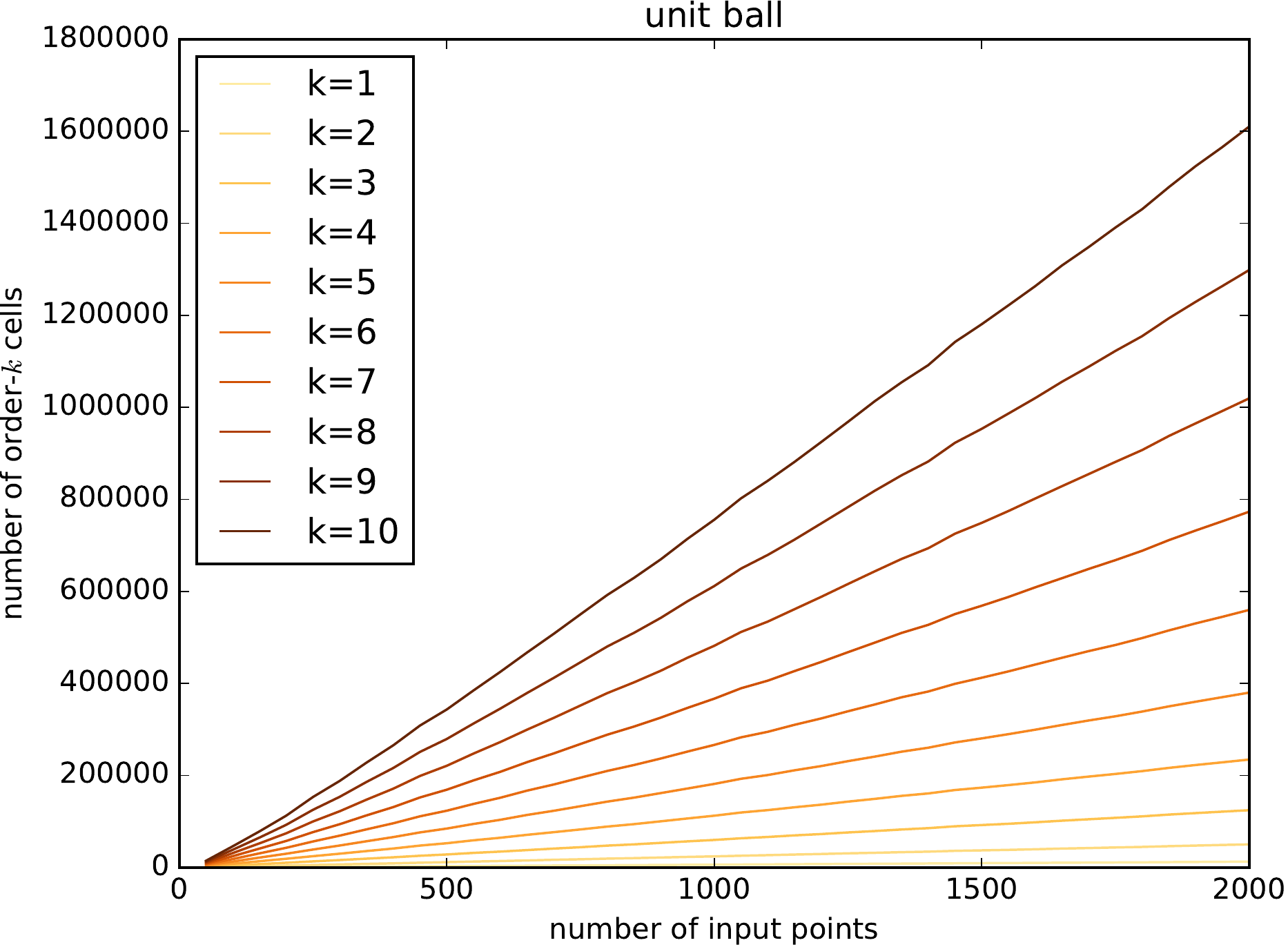}
  \end{subfigure}
  ~ 
  \begin{subfigure}[t]{0.48\textwidth}
      \centering
      \includegraphics[width=\textwidth]{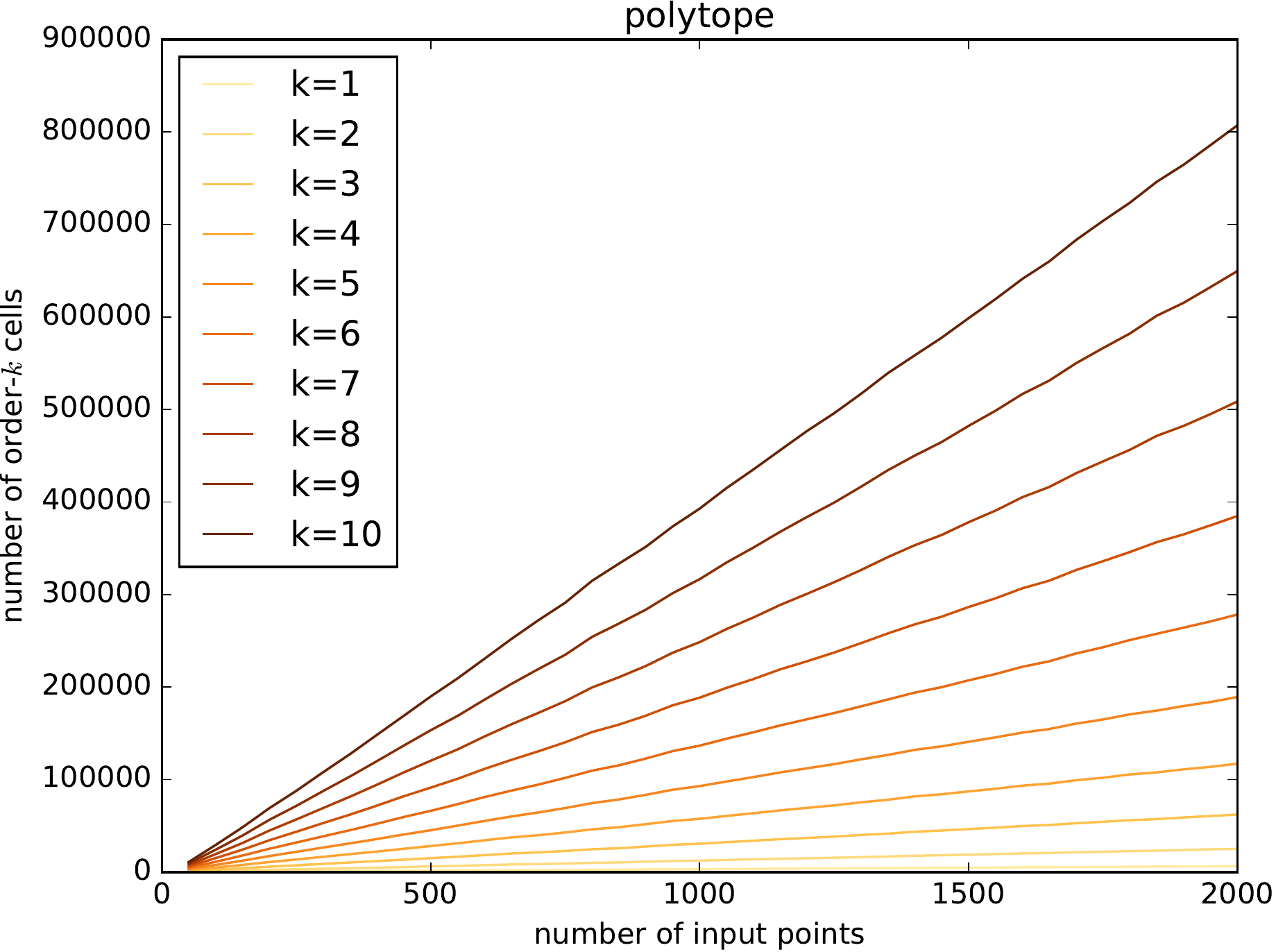}
  \end{subfigure}
  \caption{Number of cells in the order-$k$ Delaunay mosaics
    for small $k$ in relation to the input size,
    for various 3-dimensional point sets.}
  \label{fig:complexity}  
\end{figure}

\medskip \noindent
\emph{Variance.}
To probe whether the above figures are representative, 
we investigate the variance in number of cells
for the {\tt polytope} and the {\tt unit ball}.
As shown in Figure \ref{fig:variance},
the variance is particularly small for the {\tt polytope},
and it is considerably larger of the {\tt unit ball}.
Curiously, the variance dips at $k = n/2$.
\begin{figure}[hbt]
  \centering
  \begin{subfigure}[t]{0.48\textwidth}
    \centering
    \includegraphics[width=\textwidth]{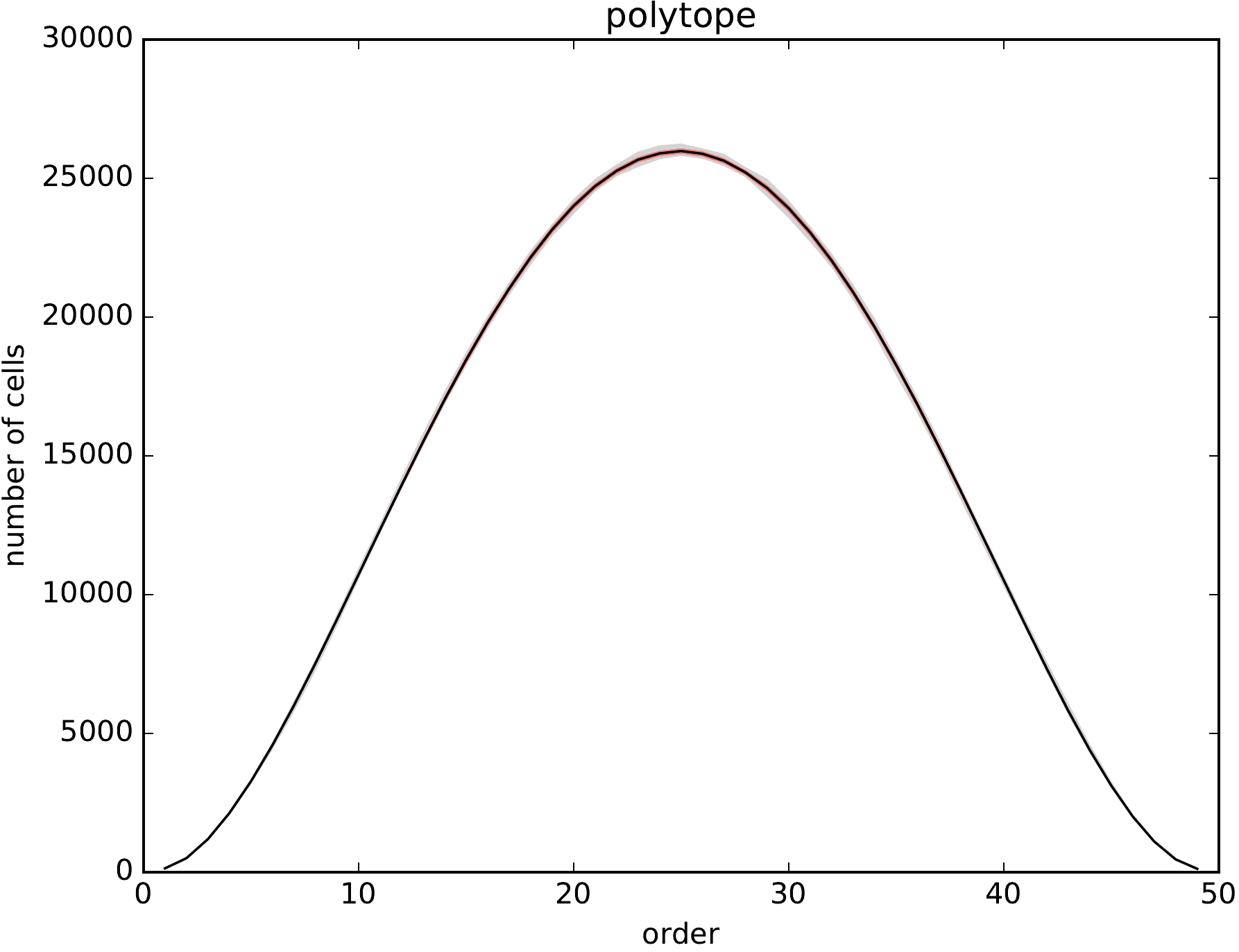}
  \end{subfigure}
  ~
  \begin{subfigure}[t]{0.48\textwidth}
    \centering
    \includegraphics[width=\textwidth]{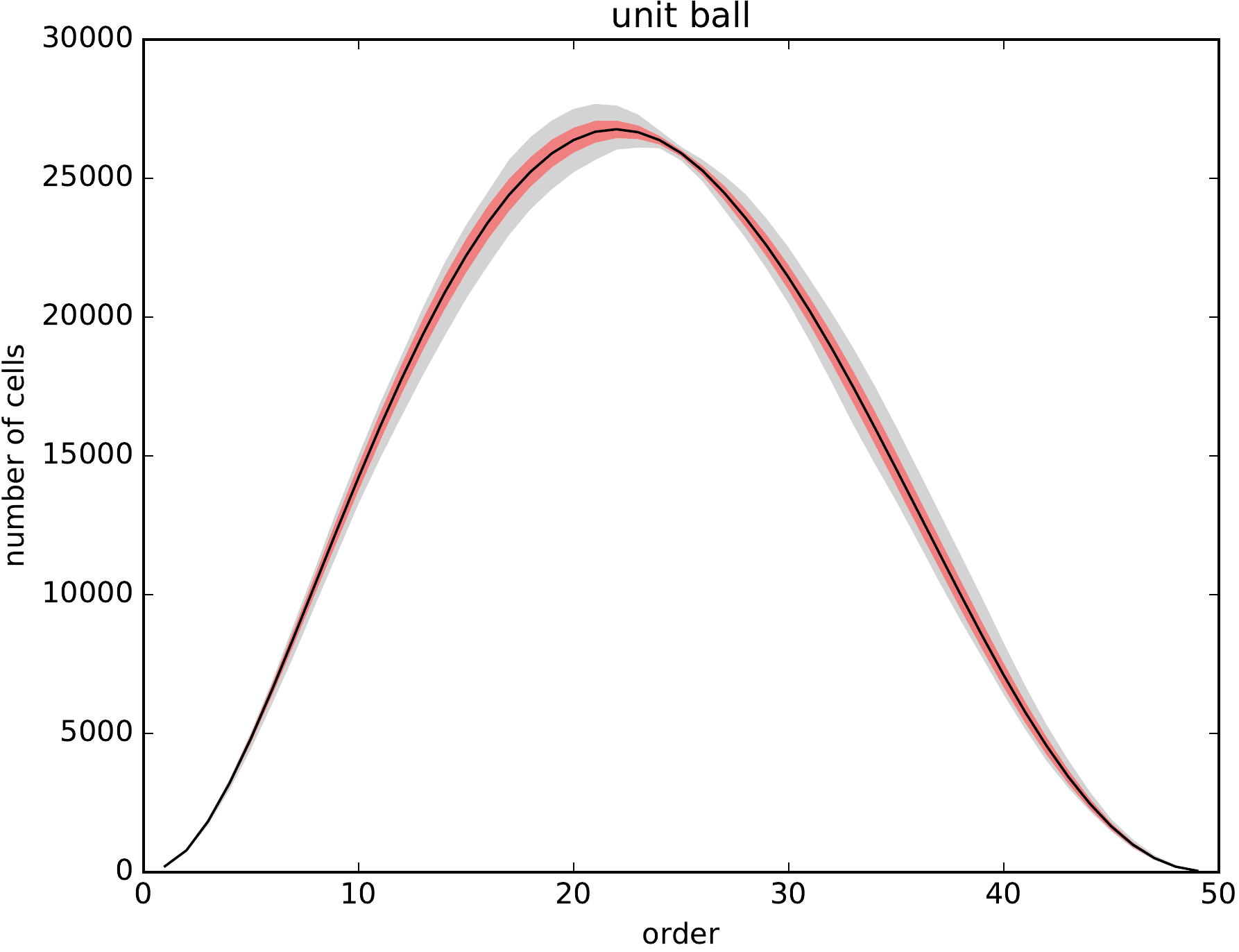}
  \end{subfigure}
  \caption{Variance of the number of $3$-dimensional cells in the
    order-$k$ Delaunay mosaics of randomly sampled points in convex position
    (\emph{left}) and in a unit ball (\emph{right}).
    The statistics of each plot are obtained from $30$ sets of $50$ points each.
    In \emph{black}: the mean;
    in \emph{red}: the range of one standard deviation around the mean;
    in \emph{grey}: the range between the minimum and maximum.}
  \label{fig:variance}
\end{figure}
\medskip \noindent

\medskip \noindent
\emph{Generations.}
We also investigate the distribution of cells of different generations.
All point sets exhibit a pattern similar to that in
Figure \ref{fig:generations_distribution},
with the fraction of first-generation cells decreasing
and the fraction of $d$-th-generation cells increasing as the order grows.
The change is most prominent for small and large $k$,
while the fractions remain almost constant in the range $k \approx n/2$,
provided $n$ is significantly larger than the dimension $d$.
\begin{figure}[hbt]
  \centering
  \begin{subfigure}[t]{0.48\textwidth}
    \centering
    \includegraphics[width=\textwidth]{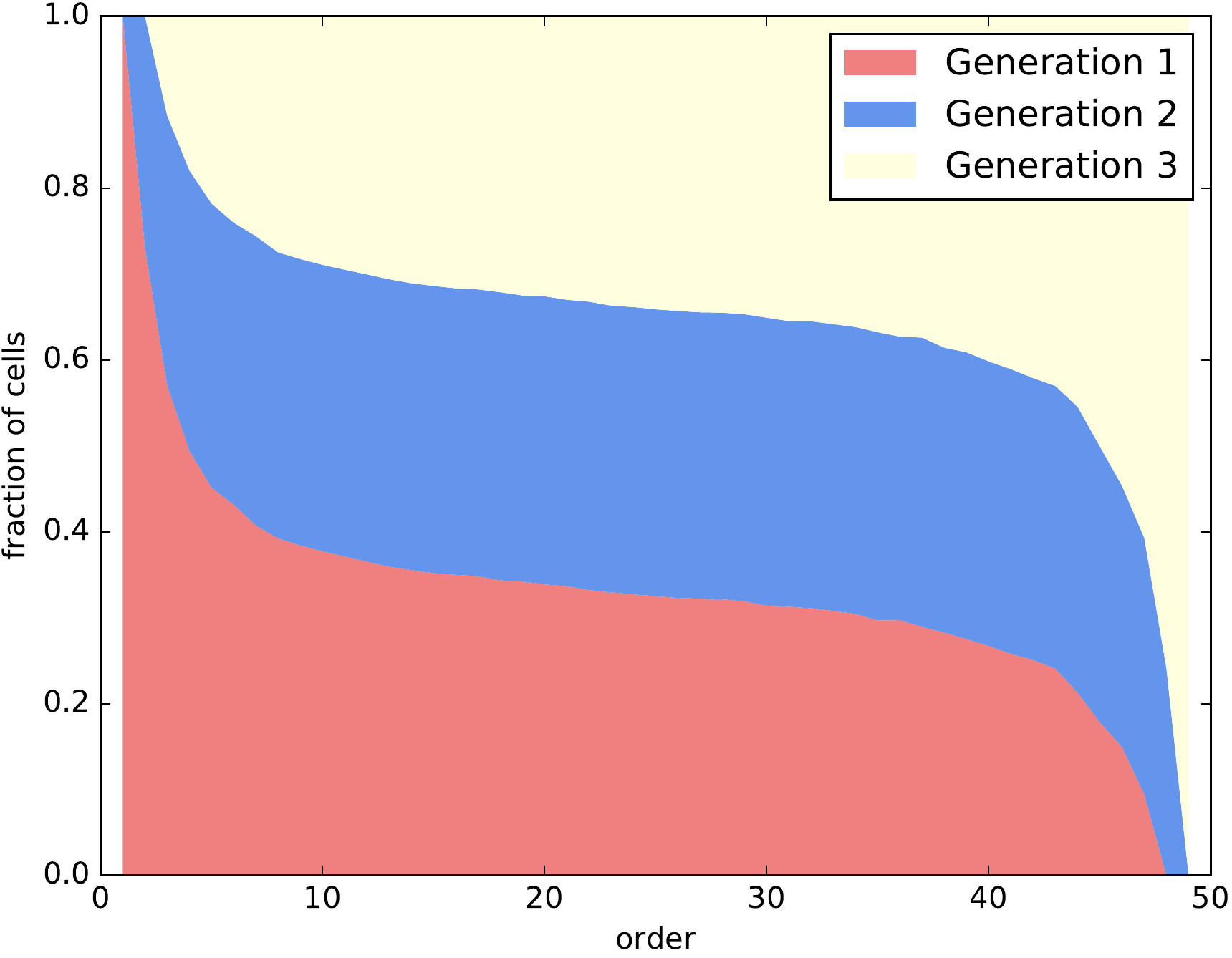}
  \end{subfigure}
  ~
  \begin{subfigure}[t]{0.48\textwidth}
    \centering
    \includegraphics[width=\textwidth]{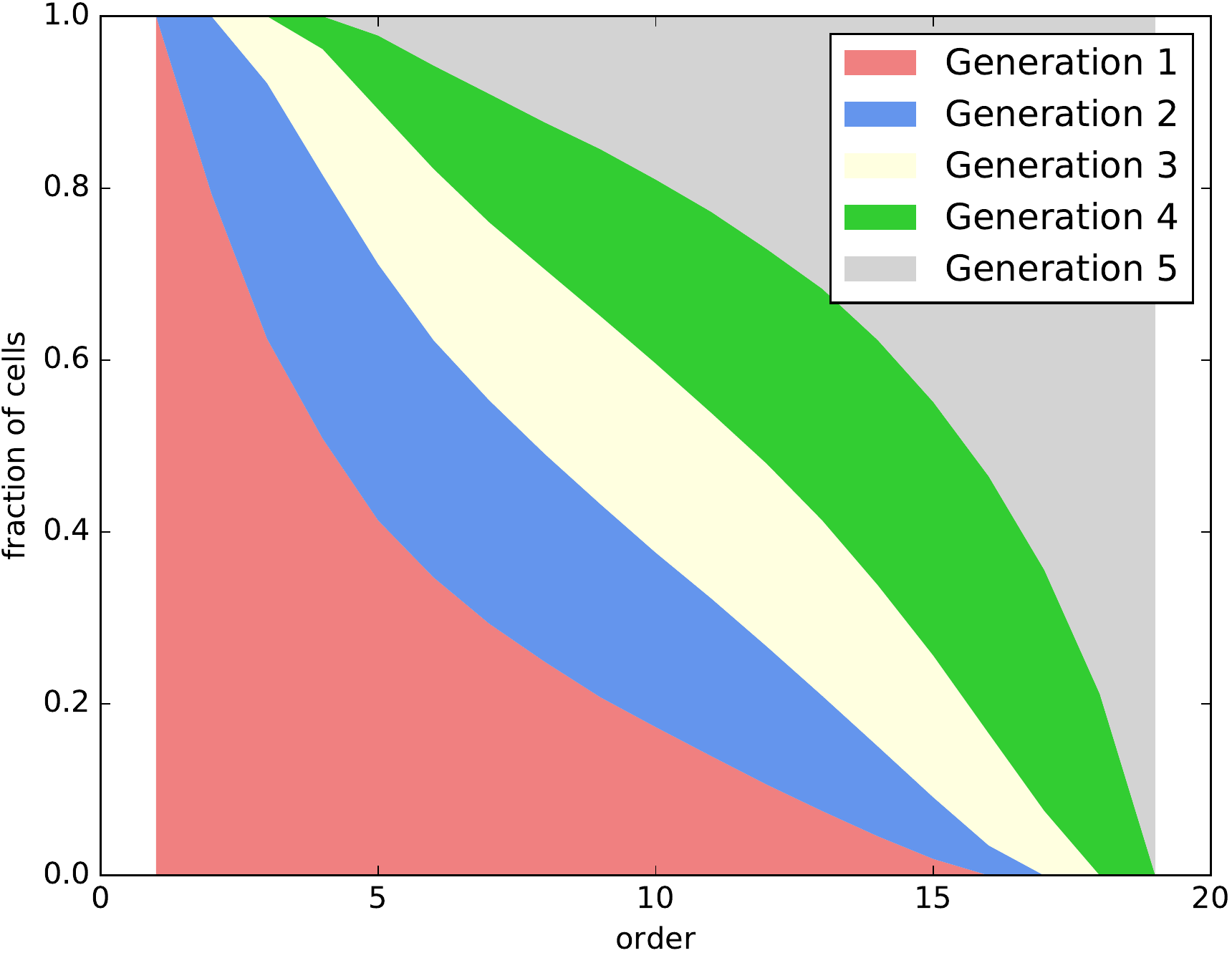}
  \end{subfigure}
  \caption{Fraction of cells of each generation in the order-$k$ Delaunay
    mosaic, for $50$ random points in the unit $3$-ball (\emph{left})
    and $20$ random points in the unit $5$-ball (\emph{right}).}
  \label{fig:generations_distribution}
\end{figure}

\medskip \noindent
\emph{Curse of dimensionality.}
Like many geometric structures, order-$k$ Delaunay mosaics are subject
to the dimensionality curse.
Figure \ref{fig:dimension_complexity} shows how the size of order-$k$
Delaunay mosaics behaves for point sets in different dimensions. 
\begin{figure}[hbt]
  \centering
  \begin{subfigure}[t]{0.48\textwidth}
    \centering
    \includegraphics[width=\textwidth]{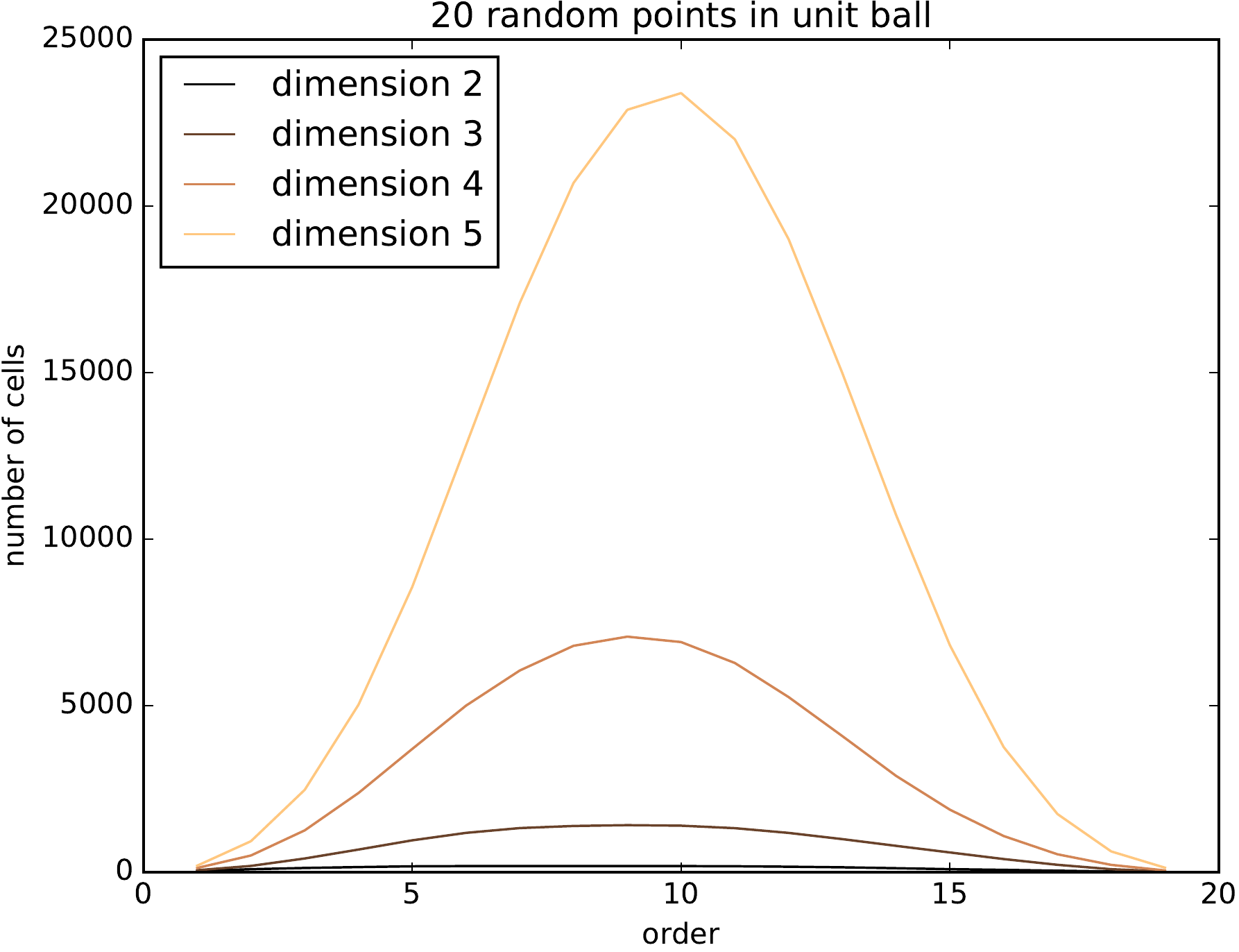}
  \end{subfigure}
  ~ 
  \begin{subfigure}[t]{0.48\textwidth}
    \centering
    \includegraphics[width=\textwidth]{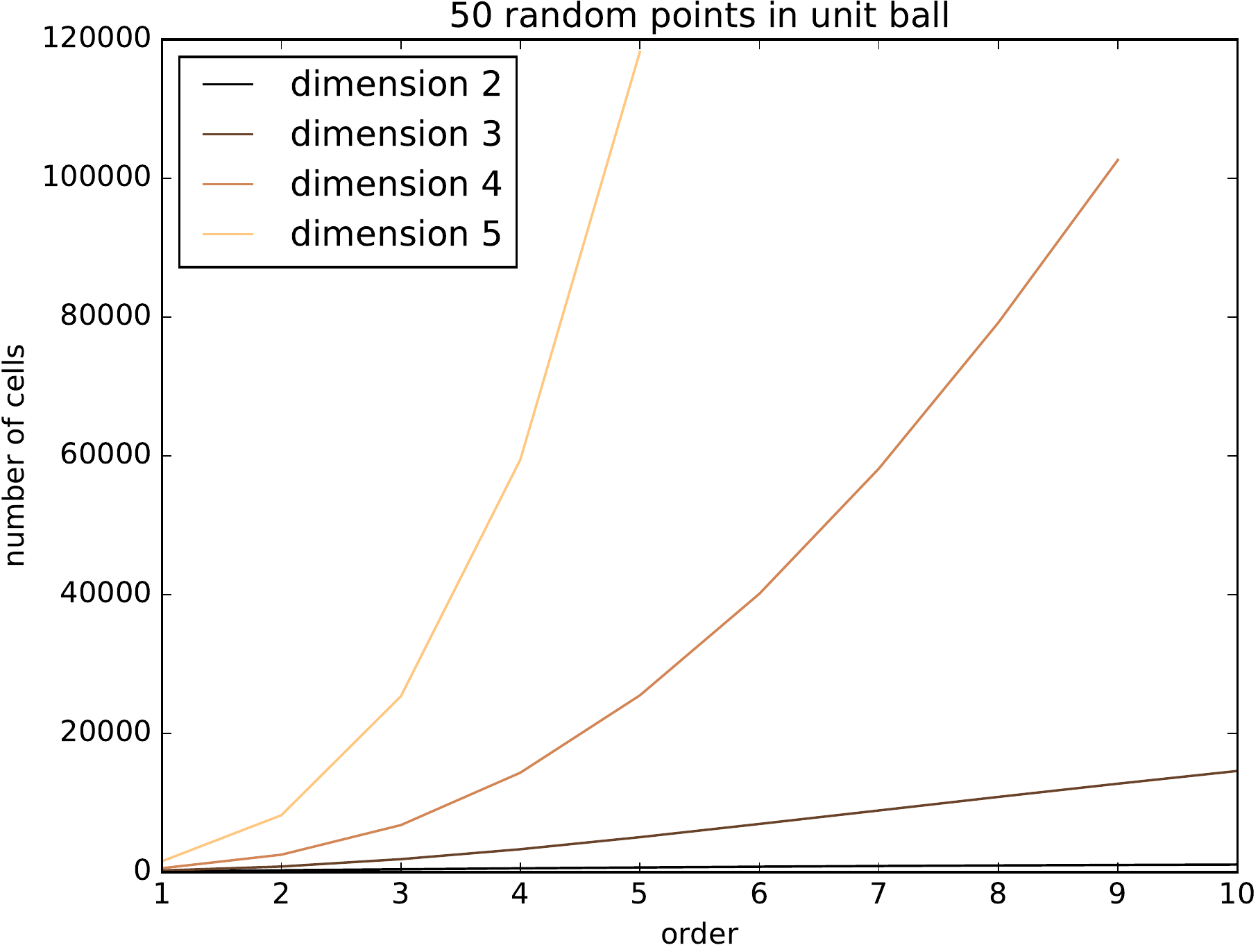}
  \end{subfigure}
  \caption{Number of $d$-cells in the order-$k$ Delaunay mosaic
    for $20$ points (\emph{left}) and $50$ points (\emph{right})
    randomly sampled in the unit ball for different dimensions $d$.}
  \label{fig:dimension_complexity}  
\end{figure}

\medskip \noindent
\emph{Vertex degrees.}
Order-$k$ Delaunay mosaics in $\Rspace^3$
exhibit an interesting distribution of vertex
degrees for random point sets; see Figure \ref{fig:degree_distribution}.
The distribution looks like the sum of two distributions---with the second
one only covering values $2$ modulo $3$---and is exhibited for all $k$ except
very small and very large ones.
We do not know the reason for vertices being frequently incident
to $5, 8, 11, \dots$ $d$-cells,
but suspect these numbers correspond to geometric configurations
of cells of different generations,
such as three octahedra sharing a common vertex with two tetrahedra.
\begin{figure}[hbt]
  \centering
  \begin{subfigure}[t]{0.48\textwidth}
    \centering
    \includegraphics[width=\textwidth]{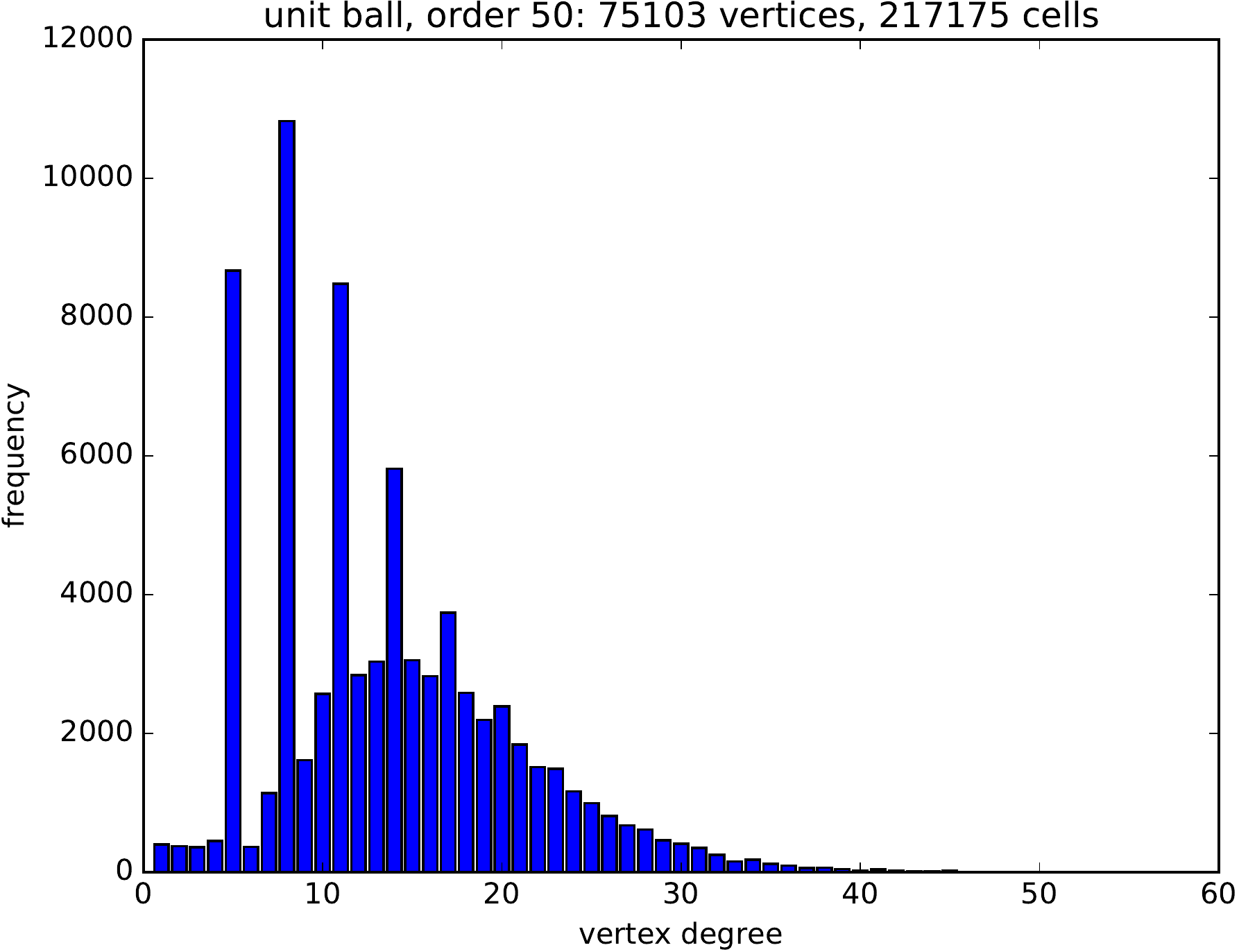}
  \end{subfigure}
  ~
  \begin{subfigure}[t]{0.48\textwidth}
    \centering
    \includegraphics[width=\textwidth]{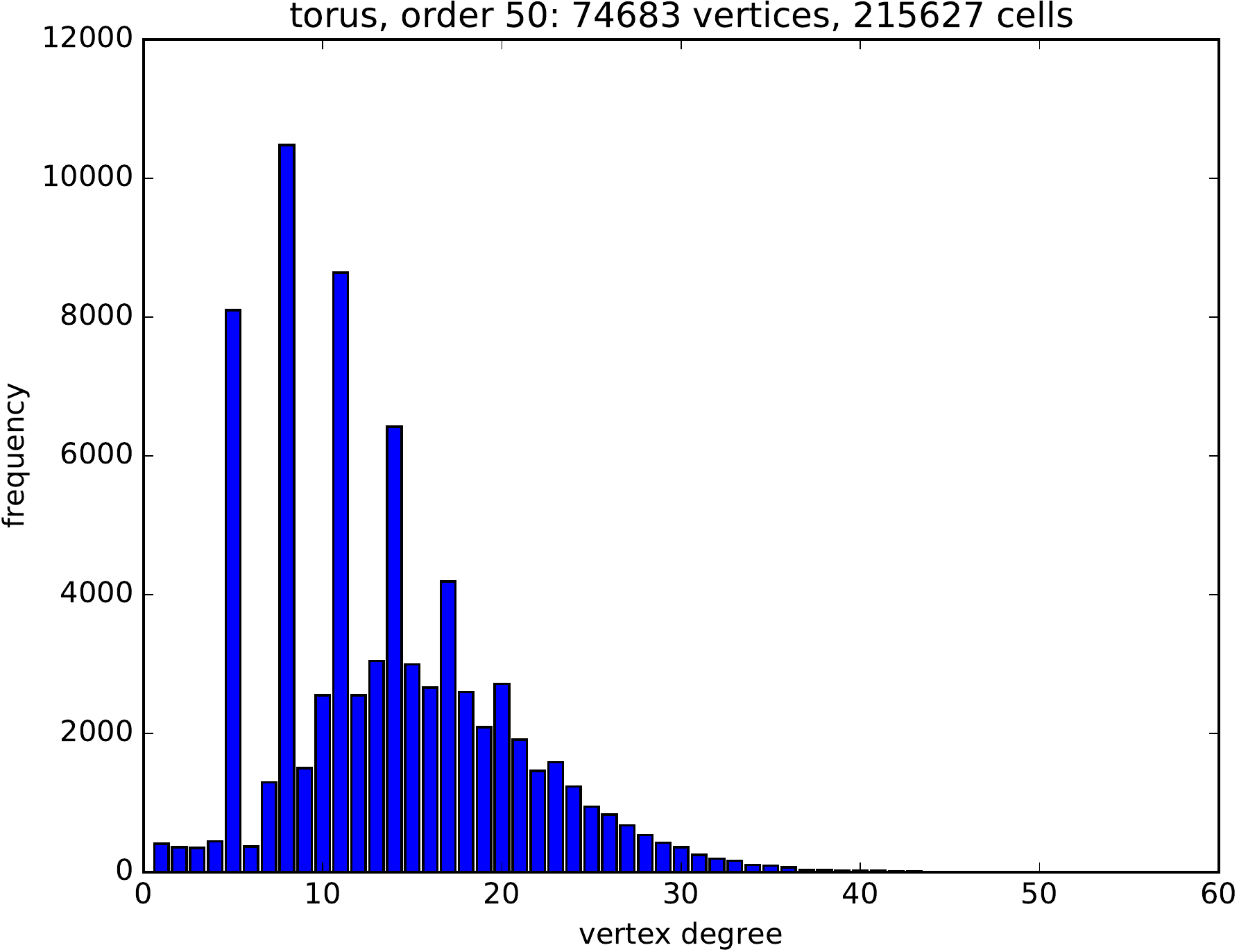}
  \end{subfigure}
  \caption{Vertex degree distribution in the order-$50$ Delaunay
    mosaic for $100$ points sampled in the unit ball (\emph{left})
    and on the torus (\emph{right}).}
  \label{fig:degree_distribution}
\end{figure}

\medskip \noindent
\emph{Clusters.}
First-generation cells of any order-$k$ Delaunay mosaic come in
clusters connected by shared facets.
We investigate the distribution of their sizes, leaving the discussion of
their potential algorithmic significance for later.
Figure \ref{fig:cluster_sizes} shows cluster size distributions
in $\Rspace^3$ for different orders.
For very small $k$, the distribution depends on how the points
are sampled, while for all other $k$, the cluster sizes seem
to follow an exponential distribution.
The decay rate increases with $k$ and seems to be linked to the fraction of
first-generation cells.
It culminates in all clusters being singletons for $k = n-3$.
For $k > n-3$, there are no more first-generation cells.
\begin{figure}[hbt]
  \centering
  \begin{subfigure}[t]{0.31\textwidth}
    \centering
    \includegraphics[height=3.6cm]{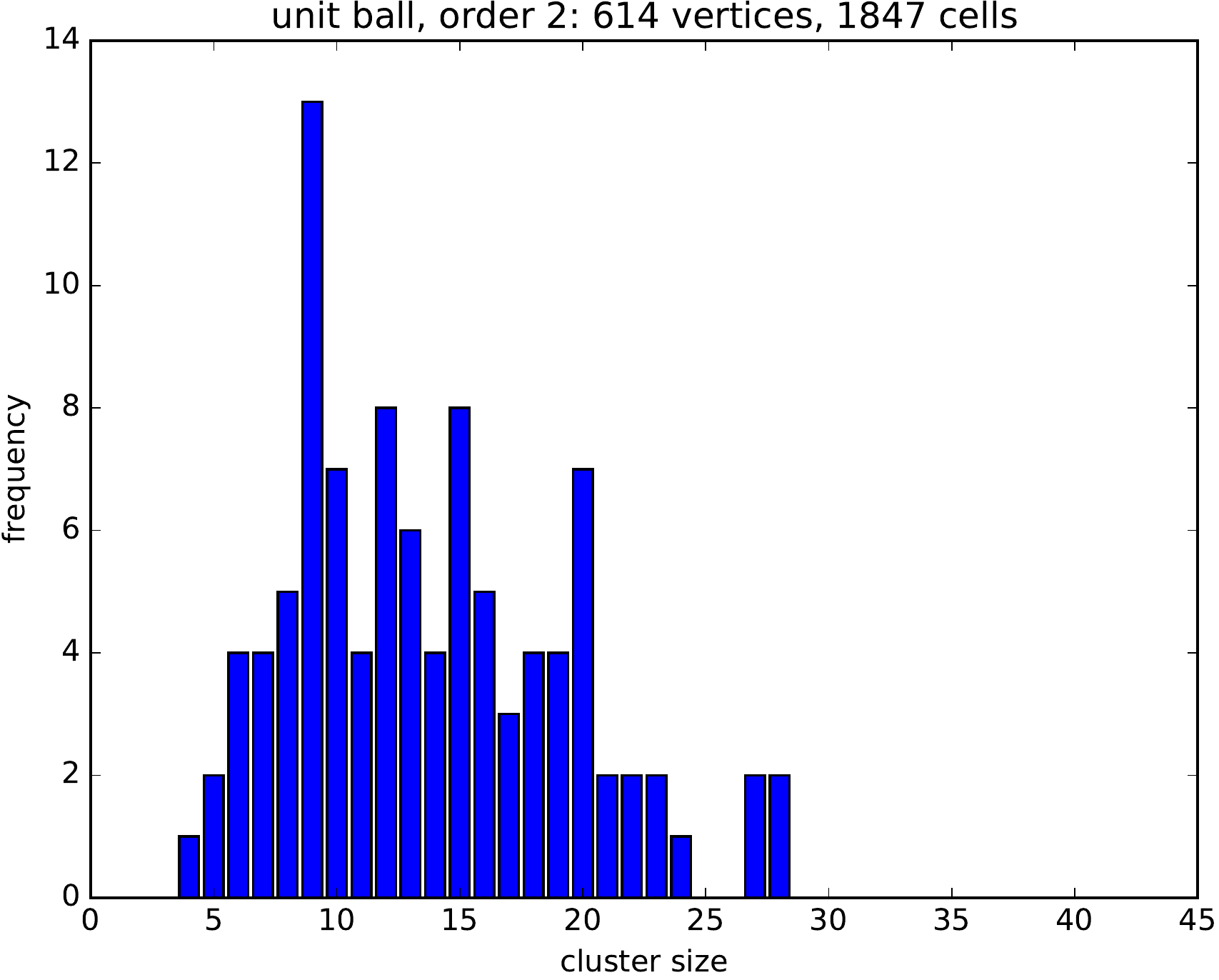}
  \end{subfigure}
  ~ 
  \begin{subfigure}[t]{0.31\textwidth}
    \centering
    \includegraphics[height=3.6cm]{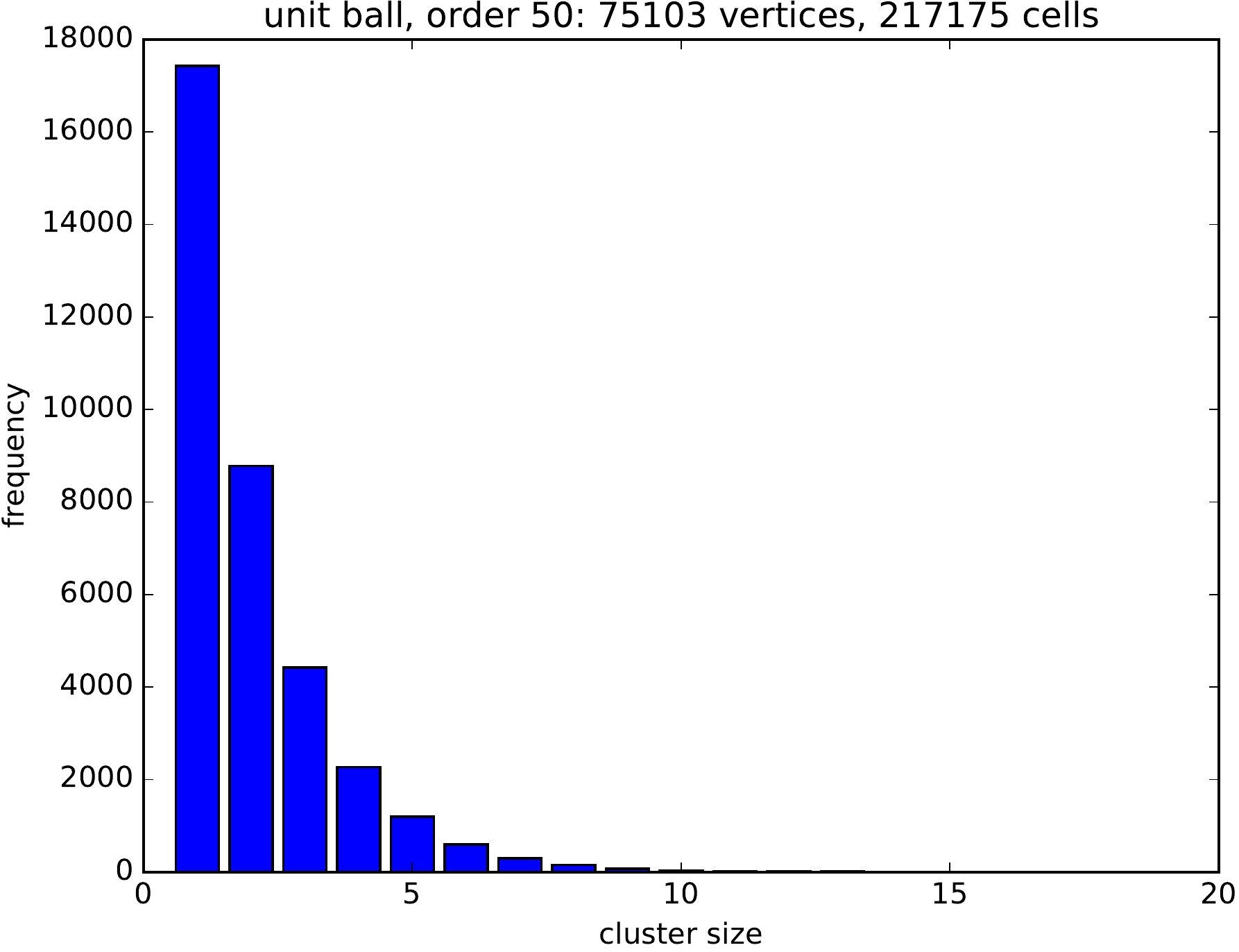}
  \end{subfigure}
  ~ 
  \begin{subfigure}[t]{0.31\textwidth}
    \centering
    \includegraphics[height=3.6cm]{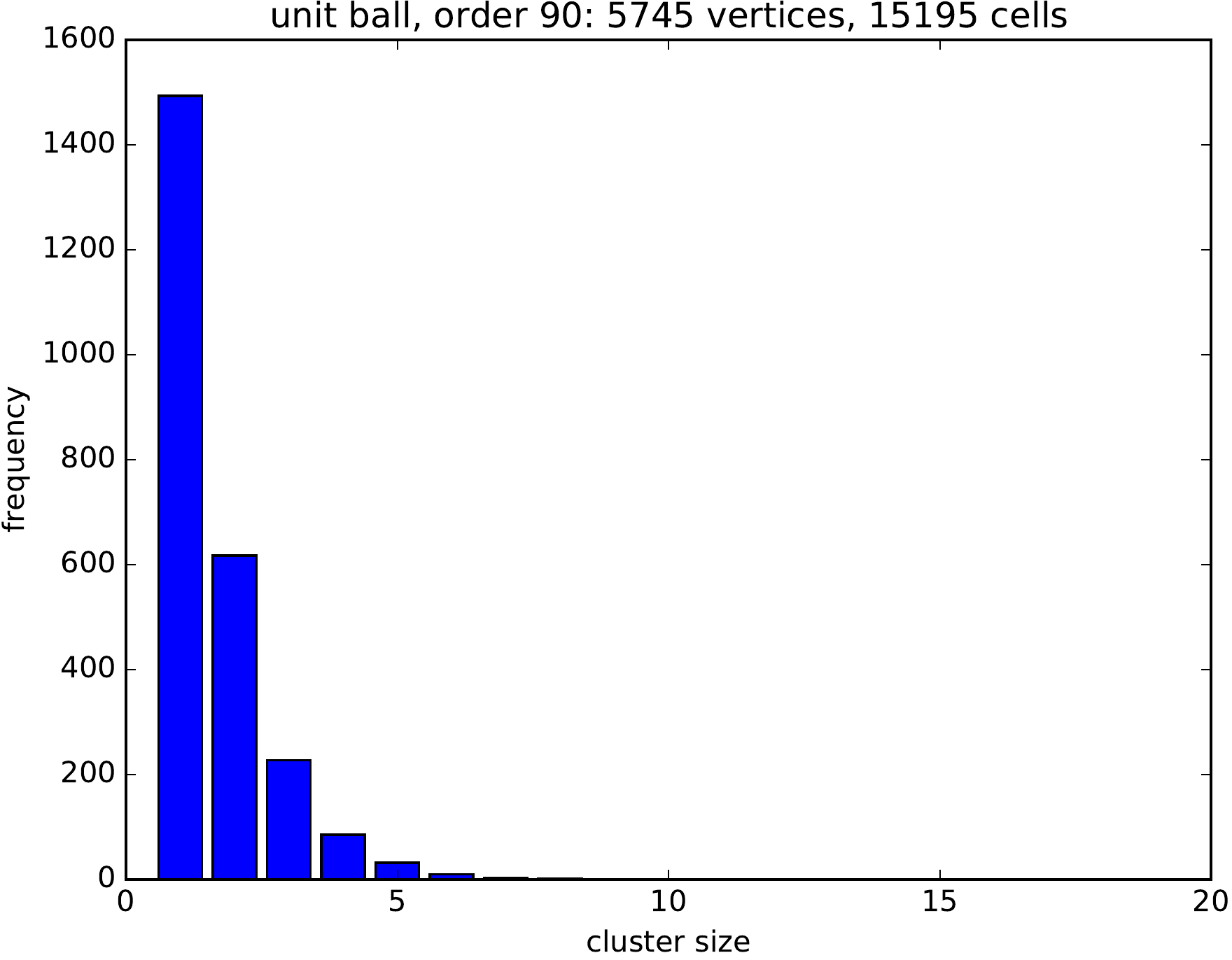}
  \end{subfigure}
  \caption{From \emph{left} to \emph{right}: distribution of cluster sizes
    in Delaunay mosaics of order $2$, $50$, and $90$
    for $100$ random points in the {\tt unit ball}.}
  \label{fig:cluster_sizes}  
\end{figure}

\section{Order-$k$ Alpha Shapes}
\label{sec:alphashapes}  

Beyond order-$k$ Delaunay mosaics, our algorithm can be extended
to compute \emph{order-$k$ alpha shapes}, as introduced in \cite{KrPo14}.
To this end, the rhomboid tiling
is endowed with a radius function on its rhomboids \cite{EdOs18}.
It is inherited by the Delaunay mosaics, which are slices of the rhomboid
tiling, and their sublevel sets with respect to this radius function
are complexes that geometrically realize the order-$k$ $\alpha$-shapes.
In this section, we recall the definition of the radius function
from \cite{EdOs18}, and present an efficient way of computing it.

\Skip{While the algorithm from Section \ref{sec:algorithm}
only explicitly computes top-dimensional rhomboids,
this radius function is defined on rhomboids of all dimensions.
The $3^\card{\xin{\rho}}$ faces of each rhomboid, $\rho \in \Rhomboid{\Xp}$, are easy to obtain: 
Each 3-partition of $\xon{\rho} = \Xp_I \sqcup \Xp_O \sqcup \Xp_U$ gives rise
to a face $\rho'$ of $\rho$ with 
$\xin{\rho'} = \xin{\rho} \cup \Xp_I$, 
$\xon{\rho'} = \Xp_O$ and 
$\xout{\rho'} = \xout{\rho} \cup \Xp_U$. 
Thus we obtain the rhomboids of all dimensions in the rhomboid tiling
by enumerating the partitions of $\xon{\rho}$ for each top-dimensional rhomboid.}

To get started, we note that the radius function needs
a representation for every rhomboid in the tiling, but the algorithm
in Section \ref{sec:algorithm} computes only the top-dimensional rhomboids.
This is easily remedied by noticing that the dimension of a rhomboid is
$k = \card{\xon{\rho}}$ and its $3^k$ faces correspond to the different
ways of partitioning $\xon{\rho}$ into three sets.
For the remainder of the discussion, assume that we have a representation
for the rhomboids of all dimensions $0 \leq j \leq d+1$ in $\Rhomboid{\Xp}$.
Each $j$-dimensional rhomboid, $\rho \in \Rhomboid{\Xp}$,
corresponds to a $(d+1-j)$-dimensional cell in the dual arrangement,
$\rho^* \in \Arr{\Xp}$.
We introduce $\parafunct{t} (\rp) \colon \Rspace^d \to \Rspace$ defined
by mapping $x \in \Rspace^d$ to
$\parafunct{t} (\rp) = \tfrac{1}{2} (\norm{\rp}^2 - t)$.
With slight abuse of notation, we write $\Paraboloidt{t}$ 
for the graph of this function.
This graph is the paraboloid $\Paraboloid_0$
dropped down vertically by a distance $\tfrac{t}{2}$.
We define the \emph{squared radius function}
$\Rfun^2 \colon \Rhomboid{\Xp} \to \Rspace$,
which maps a rhomboid to the minimum $t$ such that
$\Paraboloidt{t}$ has a non-empty intersection with $\rho^*$.
We call a sphere \emph{constrained} by $\rho$ if it encloses $\xin{\rho}$,
passes through all points of $\xon{\rho}$, and has no other points of $A$ inside.
Letting $\Smin{\rho}$ be the smallest such sphere,
we get an alternative interpretation of the radius function:
\begin{lemma} \label{lem:rhomboid_radius}
  $\Rfun^2(\rho)$ equals the squared radius of $\Smin{\rho}$.
\begin{proof}
  The proof of Theorem~1 of \cite{EdOs18} establishes a map from points
  in the $\Arr{\Xp}$ to spheres:
  a point $\rpd = (\rp, \rpcd) \in \Rspace^d \times \Rspace$
  below the paraboloid $\Paraboloid$ is mapped to the sphere, $S$,
  with center $\rp$ and squared radius $\norm{\rp}^2 - 2\rpcd$.
  Importantly, if $\rho^*$ is the cell in the dual arrangement whose interior
  contains $\rpd$, then $S$ is constrained by $\rho$,
  which is the rhomboid dual to $\rho^*$.

  Now let $t = \Rfun^2(\rho)$,
  and let $r^2$ be the squared radius of $\Smin{\rho}$.
  By definition, $t$ is the smallest value for
  which $\Paraboloidt{t}$ contains a point $\rpd \in \rho^*$.
  The aforementioned map maps $\rpd$ to a sphere constrained by $\rho$,
  thus $r^2 \leq t$.
  When reversing this map, $\Smin{\rho}$ is mapped to a point of $\rho^*$.
  As $t$ was the smallest value for which $\Paraboloidt{t}$
  touches $\rho^*$, we have $t \leq r^2$.
  Thus the squared radius of $\Smin{\rho}$ equals $\Rfun^2(\rho)$.
\end{proof}
\end{lemma}

To compute this radius function, we first get the smallest
sphere constrained by a rhomboid.
While Welzl's algorithm \cite{Welzl91smallestenclosing}
for smallest enclosing sphere can be adapted to this task, it takes
$O(n)$ with $n = \card{\Xp}$ for each such sphere computation.
To improve on this bound, we recall that Lemma~3 of \cite{EdOs18}
establishes that rhomboids of the same
radius value come in intervals $[\rho_{min}, \rho_{max}] := 
\{\rho \in \Rhomboid{\Xp} \mid \rho_{min} \subseteq \rho \subseteq \rho_{max}\}$
whose lower bound, $\rho_{min}$, is a vertex.
To identify the vertex $\vq$ that a rhomboid $\rho$ forms an interval
with, we need to identify its vertex with the same radius value.
By Lemma~\ref{lem:rhomboid_radius} this means the radii of $\Smin{\rho}$
and $\Smin{\vq}$ have to be the same, and it is not difficult to see
that the spheres $\Smin{\rho}$ and $\Smin{\vq}$ are in fact the same.
As $\xon{\vq} = \emptyset$ for any vertex $\vq$, the sphere achieving the radius value
of $\vq$ is defined solely by inclusions and exclusion constraints.
Therefore all constraints of $\rho$ that require points of $\xon{\rho}$
to be on the sphere need to be converted to inclusion and exclusion
constraints without affecting the resulting sphere. We know that such
constraints exist because the lower bound of the interval is a vertex.
This observation gives rise to the following lemma.

\begin{lemma}
Let $\rho$ be a rhomboid that is an upper bound of an interval.
Let $\Xp_I \subseteq \xon{\rho}$ such that the smallest enclosing sphere $S$
of $\Xp_I$ that excludes $\xon{\rho} \setminus \Xp_I$ is the same as
the circumsphere of $\xon{\rho}$.
Then $\rho$ forms an interval with the vertex $\vq = \xin{\rho} \cup \Xp_I$.
\begin{proof}
As $\rho$ is an upper bound of an interval, its sphere, $\Smin{\rho}$,
is only supported by $\xon{\rho}$.
Indeed, if there were another point 
$\xp \in \xin{\rho}$---or $\xp \in \xout{\rho}$---on the surface of this
sphere, then the rhomboid $\varrho$ with $\xon{\varrho} = \xon{\rho} \cup \{\xp\}$
and $\xin{\varrho} = \xin{\rho} \setminus \{\xp\}$---or
$\xout{\varrho} = \xout{\rho} \setminus \{\xp\}$---would be
a higher-dimensional rhomboid with the same
sphere $\Smin{\varrho} = \Smin{\rho}$,
contradicting that $\rho$ be an upper bound of an interval.

As $\Smin{\rho}$ is only supported by $\xon{\rho}$, this means that 
$\Smin{\rho}$ is the same as the circumsphere of $\xon{\rho}$,
which by our assumption is the same as $S$.
Now the inclusion and exclusion constraints of $S$ are part of the
constraint set for $\Smin{\vq}$, but because $S = \Smin{\rho}$ it
does in fact fulfill all the constraints of $\Smin{\vq}$.
Thus $\Smin{\vq} = S = \Smin{\rho}$, proving that they are in the
same interval.
\end{proof}
\end{lemma}

\ourparagraph{Algorithm.}
Assume $\rho$ is a $j$-rhomboid that is an upper bound of an interval.
Let $S$ be the circumsphere of $\xon{\rho}$.
For each point $\xp \in \xon{\rho}$, we need to decide
whether to impose an inclusion or exclusion constraint on it.
Let $S_\xp$ be the circumsphere of $\xon{\rho} \setminus \{\xp\}$.
If $\xp$ is outside of $S_\xp$, then imposing an exclusion constraint
for $\xp$ would yield $S_\xp$ rather than $S$, thus we add $\xp$
to $\Xp_I$ in order to impose an inclusion constraint for it.
Similarly, if $\xp$ is inside of $S_\xp$, we have to impose an exclusion
constraint for $\xp$ and thus do not add it to $\Xp_I$.

While this is difficult for an individual rhomboid, it becomes
straightforward if we compute all intervals in the rhomboid tiling.
We know that all $(d+1)$-rhomboids are upper bounds of intervals.
After marking all rhomboids that are contained in such intervals, we know
that all remaining unmarked $d$-rhomboids are upper bounds of intervals.
Thus by processing the rhomboids in decreasing dimension, all unmarked
rhomboids we encounter are upper bounds.

\section{Discussion}
\label{sec:discussion}  

This paper presents a simple algorithm for computing order-$k$ Delaunay mosaics
in Euclidean space of constant dimension.
Implementations of the algorithm---in {\tt C++} for points in $\Rspace^2$ and $\Rspace^3$
and in {\tt python} for points in $\Rspace^d$---are provided \cite{orderkgithub,rhomboidgithub}.
This software includes the application to the persistence of $k$-fold covers
described in \cite{EdOs18}.
The remainder of this section discusses this application and possible extensions
and optimizations of our algorithm.

\medskip \noindent
\emph{$k$-fold covers.}
The sublevel sets of the order-$k$ Delaunay mosaics with respect to the
radius function introduced in Section \ref{sec:alphashapes}
are homotopy equivalent to $k$-fold covers of Euclidean balls.
It follows that our algorithms facilitate the computation of persistence
of these $k$-fold covers.
Furthermore, the circumcenters of the spheres that are used in the computation of the
radius function provide the geometric locations of the
order-$k$ Voronoi vertices and allow reconstructing the order-$k$ Voronoi
tessellation via duality.

\medskip \noindent
\emph{Weighted setting.}
Our algorithm generalizes to points with real weights, but not easily.
The main challenge is the extraction of the vertices of the order-$k$ mosaic
from lower-order mosaics.
This extraction relies on Theorem \ref{thm:orderk-vertices},
which does not hold for weighted points.
Indeed, a crucial assumption in this theorem is that every lifted hyperplane
is incident to the depth-$0$ chamber of the arrangement,
and this property is generally violated for weighted points.
This is the same assumption used in the prior dimension-agnostic algorithms
\cite{agarwal1998constructing,Mul90,mulmuley1991levels}.
For sets of weighted points that satisfy this assumption,
our algorithm and these prior algorithms still work.
To overcome this limitation, we would need a way to detect all bowls
in the arrangement, because they correspond to the vertices in the
Delaunay mosaics our algorithm is not able to find.
Identifying the bowls is an independent problem,
and any solution to it can be combined with our algorithm.
Once we know the bowls and add the corresponding vertices
to the appropriate mosaics, our algorithm works as before.

\medskip \noindent
\emph{Clusters of cells.}
As mentioned in Section \ref{sec:experimental}, first-generation cells in
the order-$k$ Delaunay mosaic are organized in clusters.
To formally define them, consider the graph whose nodes are the cells
and whose arcs are the shared facets (i.e.\ the $1$-skeleton of the
order-$k$ Voronoi tessellation).
A \emph{cluster} is a connected component in the subgraph induced
by the first-generation cells.
It is not difficult to see that two such cells belong to a common cluster
if and only if the corresponding rhomboids have the same anchor vertex.
Let $\rho$ be one of these rhomboids and recall that the anchor vertex is $\Ain(\rho)$,
which in this case is a collection of $k-1$ points of $A$.
Each combinatorial vertex of any cell in the cluster contains
these $k-1$ points, plus one additional point,
which differentiates between these vertices.
In other words, the cluster as a subcomplex of the
order-$1$ Delaunay mosaic of these additional points.

With this insight, we could replace the weighted Delaunay mosaic of
the entire vertex set by multiple instances of
unweighted Delaunay mosaics, namely one per cluster.
This alternative strategy avoids the need to compute averages of points
at the cost of extra book-keeping to group the vertex set of
$\Del{k}{A}$ into clusters.
We mention that in $\Rspace^2$, the structure of each cluster satisfies
the requirements that allow for the construction in time linear in
the number of points \cite{AgSh89}.

\medskip \noindent
\emph{Exact arithmetic.}
The CGAL software library \cite{cgal:eb-19b} supports exact arithmetic
by distinguishing between \emph{exact constructions} and \emph{exact predicates}.
The latter are geometric tests with a {\tt true} or {\tt false} answer,
such as whether or not a given point lies on a given sphere.
By itself, the CGAL algorithm for weighted Delaunay triangulations
requires exact predicates but no exact constructions.
Our algorithm, on the other hand, computes averages of collections of
input points, which are the locations of the vertices of the mosaic.
This is an exact construction and indeed the only one needed to run
our algorithm with exact arithmetic.
In practice, exact constructions are a significant overhead with noticeable
impact on the runtime, which would be nice to avoid.

\newpage
\bibliography{he_algo}{}
\bibliographystyle{plain}

\end{document}